\documentclass[11pt]{article}

\usepackage{color,fullpage,url} 
\usepackage[utf8]{inputenc}

\usepackage[colorlinks=true, linkcolor=blue, urlcolor=blue, citecolor=blue, backref=page]{hyperref}
\usepackage{tikz}
\usepackage{authblk}
\usepackage{lineno}
\usepackage{bm}

\usepackage{cite} 

\usepackage{graphicx}
\usepackage{booktabs}
\usepackage{amsthm}
\usepackage{amssymb} 
\usepackage{mathrsfs}
\usepackage{amsmath}
\usepackage{cleveref}
\usepackage[title]{appendix} 
\usepackage{titlesec} 

\usepackage[ruled,linesnumbered]{algorithm2e}
\newtheorem{thm}{Theorem}

\newtheorem{lem}{Lemma}[section]
\newtheorem{prop}[thm]{Proposition}

\theoremstyle{definition}

\title{Local Quantum Search Algorithm for Random $k$-SAT with $\Omega(n^{1+\epsilon})$ Clauses}
\author[1]{Mingyou Wu}

\begin{document}
	\pagenumbering{gobble}
	\maketitle
	\begin{abstract}
		 The random $k$-SAT instances undergo a ``phase transition'' from being generally satisfiable to unsatisfiable as the clause number $m$ passes a critical threshold, $r_k n$. This causes a drastic reduction in the number of satisfying assignments, shifting the problem from being generally solvable on classical computers to typically insolvable. Beyond this threshold, it is challenging to comprehend the computational complexity of random $k$-SAT. In quantum computing, Grover's search still yields exponential time requirements due to the neglect of structural information. 
		
		Leveraging the structure inherent in search problems, we propose the $k$-local quantum search algorithm, which extends quantum search to structured scenarios. Grover’s search, by contrast, addresses the unstructured case with $k=n$. Given that the search algorithm necessitates the presence of a target, we specifically focus on the problem of searching the interpretation of satisfiable instances of $k$-SAT, denoted as max-$k$-SSAT. If this problem is solvable in polynomial time, then $k$-SAT can also be solved within the same complexity. We demonstrate that, for small $k \ge 3$, any small $\epsilon>0$ and sufficiently large $n$:
		\begin{itemize}
			\item the $k$-local quantum search achieves general efficiency on random instances of max-$k$-SSAT with $m=\Omega(n^{2+\delta+\epsilon})$ using $\mathcal{O}(n)$ iterations, and 
			\item the $k$-local adiabatic quantum search enhances the bound to $m=\Omega(n^{1+\delta+\epsilon})$ within an evolution time of $\mathcal{O}(n^2)$.
		\end{itemize}
		 In both cases, the circuit complexity of each iteration is $\mathcal{O}(n^k)$, and the efficiency is assured with overwhelming probability $1 - \mathcal{O}(\mathrm{erfc}(n^{\delta/2}))$.
		By modifying this algorithm capable of solving all instances of max-$k$-SSAT, we further prove that max-$k$-SSAT is polynomial on average when $m=\Omega(n^{2+\epsilon})$ based on the average-case complexity theory.
	\end{abstract}
	
	\newpage
	\tableofcontents
	\newpage
	\pagenumbering{arabic}
	\setcounter{page}{1}
	
	\section{Introduction}
	
	The \textit{Boolean satisfiability} (SAT) problem asks whether a given Boolean formula has an interpretation. In the context of $k$-SAT problem, the Boolean formula is confined to conjunctive normal form, where each clause contains at most $k$ literals. With $k\ge3$, the NP-completeness of $k$-SAT signifies its intractability in the worst-case scenario, but it does not imply the problem as universally challenging for every instance~\cite{cook1971complexity, Karp1972}. This distinction prompts a discussion on the average-case complexity of random NP problems under specific distributions~\cite{Levin1986}.
	
	One widely used distribution for random instances of $k$-SAT is $F(n,m,k)$. In this model, each clause contains exactly $k$ literals, and a $k$-SAT instance on $n$ variables is generated by uniformly, independently, and with replacement selecting $m$ clauses from the entire set of $2^k C_n^k$ possible clauses~\cite{Achlioptas2006}. As $m$ varies, the computational complexity of random $k$-SAT instance does not always exhibit exponential growth in $n$. Instead, early numerical experiments revealed a phase transition phenomenon in complexity of instances~\cite{cheeseman1991, mitchell1992hard}, where these instances also shift from being mostly satisfiable to unsatisfiable. This led to the emergence of the \textit{satisfiability threshold conjecture}, which posits that for each $k \ge 2$, there exists a constant $r_k$ such that
	\begin{equation}\label{eq_theoshold_r_k}
		\lim_{n \to \infty} {\mathrm{Pr}[F(n,cn,k) \mathrm{\,is\,satisfiable}]=
			\begin{cases}
				1 \mathrm{ \quad if \,} c < r_k, \\
				0 \mathrm{ \quad if \, } c > r_k.
		\end{cases}}
	\end{equation}
	For $k=2$, this threshold is precisely $r_2 = 1$\cite{Chvatal1992, Goerdt1992}. For $k \ge 3$, earlier works established exact upper and lower bounds for $r_k$\cite{Kirousis1998, Coja2016}, given by
	\begin{equation}\label{eq_bound_rk}
		2^k\ln2 -\frac{1+\ln2}{2} - o_k(1) < r_k < 2^k\ln2 -\frac{1+\ln2}{2} + o_k(1).
	\end{equation}
	Moreover, the existence of a threshold $r_k$ for large $k \ge k_0$ has been rigorously proven~\cite{Ding2022}, where $k_0$ is an absolute constant.
	
	Given $m \le (r_k - \epsilon) n$, for any small $\epsilon$ and sufficiently large $n$, the existence of exponentially many potential interpretations can make solving random $k$-SAT instances more feasible in practice. However, beyond this established threshold, comprehending the complexity of random $k$-SAT becomes elusive in the domain of classical computing. Limited theoretical research delves into the range of $m$ far beyond $r_k n$. Generally, when $m$ surpasses $r_kn$, the Boolean formula becomes over-constrained, leading to the random instances becoming generally unsatisfiable. In such scenarios, the identification of contradictions might be more attainable~\cite{mitchell1992hard, Ohrimenko2007}. Nevertheless, despite existing with an exponentially low probability, the satisfiable instances still introduce significant complexity, thereby perpetuating the average complexity exponentially with respect to $n$. 
	
	This prompts us to shift our focus towards max-$k$-SSAT, a variant of max-$k$-SAT that only considers satisfiable instances. While max-$k$-SSAT belongs to NP, it lacks a polynomial reduction to any known NP-complete problem, suggesting that it might be inherently simpler and more feasible to solve. Moreover, Lemma \ref{lem_problem_reduction} demonstrates that if max-$k$-SSAT becomes polynomial-solvable, the $k$-SAT would also be solved within the same complexity. This connection implies that, while max-$k$-SSAT may not yet be confirmed as NP-complete, its solvability holds key implications for the broader P vs NP question. The central question we seek to address in this work is:
	\begin{center}
		\it How does the average-case complexity of random max-$k$-SSAT evolve \\
		\it when the number of clauses $m$ significantly exceeds the threshold $r_k n$.
	\end{center}
	
	In the current landscape of search algorithms, the well-known Grover's search~\cite{Grover1997} is designed for unstructured search problems, providing a quadratic speedup over classical algorithms. The Grover Oracle $O$ operates globally on $\left| x \right>$ in an unstructured manner, expressed as	
	\begin{equation}\label{eq_grover_Oracle}
		\left| x \right> \overset{O}{\longrightarrow}(-1)^{f(x)}\left| x \right>,
	\end{equation}
	where $f(x)$ is the objective function that $f(x)=1$ if $x=t$, and $f(x)=0$ otherwise. However, Grover's search becomes inefficient in the presence of structural information, as leveraging such information could greatly enhance search efficiency. Classical approaches, like stochastic local search~\cite{Selman1992, Selman1994, Hoos2015}, exploit this kind of structural information to accelerate problem-solving. Similarly, in the quantum domain, several quantum local search algorithms have been developed by applying quantum search techniques to classical structures, such as $d$-dimensional grids~\cite{Aaronson2006} or $k$-neighbor graphs~\cite{Tomesh2022}. However, these methods focus primarily on specific classical structures and remain limited when tackling NP-complete problems. 
	
	In this work, we explore the structural information inherent in the Oracle of search problem. Inspired by the $k$-local Hamiltonian problem~\cite{Kempe2006}, we define the $k$-local search problem, where $n$-local search corresponds to the unstructured search. This generalization naturally extends Grover's quantum search algorithm into $k$-local scenarios, termed $k$-local quantum search, with the case of $k=n$ recovering Grover's search.  Moreover, when $k$ is held constant, the $k$-local search problem is computationally tractable on classical computers, with a complexity of $\mathcal{O}(n)$ in terms of Oracle calls. It is precisely this simplicity that has led to the oversight of this problem in prior research, resulting in $k$-local quantum search remaining undiscovered for an extended period.
	
	While $k$-local search seem to be irrelevant to $k$-SAT, we illuminate the fact that $k$-local search represents the average-case scenario for satisfiable $k$-SAT instances. To be precise, we establish that for random $k$-SAT instances sharing the same interpretation, their normalized problem Hamiltonian, $\bar{H}_C$, converges in probability to the problem Hamiltonian $H_k$ of $k$-local search at a rate of $\mathcal{O}(m^{-1/2})$. Since the search algorithm necessitates a target, we focus on max-$k$-SSAT. We demonstrate that for a small $k$, the $k$-local quantum search also requires $\mathcal{O}(n)$ queries to address the $k$-local search. Moreover, when applied to random max-$k$-SSAT instances with $m=\Omega(n^{2+\epsilon+\delta})$, for any small $\epsilon>0$ and sufficiently large $n$, the algorithm's performance holds with overwhelming probability of $1 - \mathcal{O}(\mathrm{erfc}(n^{\delta/2}))$, where $\mathrm{erfc}(x)$ denotes the complementary error function.

	Moreover, by incorporating adiabatic quantum computation~\cite{Kadowaki1998, farhi2001QAA, Albash2018AQC} into our algorithm, we develop the $k$-local adiabatic quantum search algorithm, which reduces to the adiabatic Grover's search~\cite{Roland2002} when $k=n$. Leveraging the convergence of $\bar{H}_C$ and the efficiency of $k$-local quantum search, we prove that $k$-local adiabatic quantum search significantly reduces the required $m$ to order of $\Omega(n^{1+\epsilon+\delta})$, at the cost of increasing the evolution time to $\mathcal{O}(n^2)$. The efficiency is also guaranteed with a probability of $1 - \mathcal{O}(\mathrm{erfc}(n^{\delta/2}))$. For both algorithms, the circuit complexity of a single iteration is $\mathcal{O}(n^k)$. Numerical simulations are conducted on random instances of max-$3$-SSAT, yielding results that align with the theoretical predictions.
	
	Accordingly, by introducing Grover's search to handle the cases where $k$-local adiabatic quantum search fails, we demonstrate that when the number of clauses $m = \Omega(n^{2+\epsilon})$, the random max-$k$-SSAT problem is polynomial on average, based on the average-case complexity theory~\cite{Levin1986}. The distribution used in the analysis of random max-$k$-SSAT is the random model $F_s(n,m,k)$, which is derived from $F(n,m,k)$ by excluding unsatisfiable instances. In our analysis, we focus primarily on small constant values of $k \ge 3$, with $k$ not exceeding the order of tens. Throughout the following discussion, $k$ is assumed to be of this magnitude.
	
	\section{Results}\label{sec_res}

	This section initially defines the $k$-local search problem by introducing the concept of ``$k$-local'' to extend Grover's search problem beyond its unstructured nature. Based on Grover's framework, we develop the $k$-local quantum search algorithm as a generalization tailored to $k$-local scenarios. Given the simplicity of $k$-local search on classical computers, it is apparent that for a small constant $k$, the $k$-local quantum search should also be efficient within $\mathcal{O}(n)$ Oracle calls. Two key lemmas, Lemmas~\ref{lemma_gap_assume} and \ref{lemma_k_local_QS_assume}, follow directly from this observation, providing the foundation for the proof of the main result. The detailed proofs of these two are provided in Appendix~\ref{apsec_alg_design}. These two lemmas effectively isolate the quantum-specific aspects of the proof, allowing the remainder of the argument to focus on more accessible reasoning, thereby leading to the following main theorem:
	\begin{thm}[main theorem]\label{theo_main}
		For any small $\epsilon>0$ and sufficiently large $n$, the random max-$k$-SSAT with distribution $F_s(n,m,k)$ is polynomial on average when $m=\Omega(n^{2+\epsilon})$. 
	\end{thm}
	
	\subsection{Quantum search algorithms}\label{subsec_quan_alg}
	
	\subsubsection{From unstructured search to \textit{k}-local search}\label{subsubsec_k_local_search}

	Grover's search is specialized for the unstructured scenario with an objective function as
	\begin{equation}
		f(x) = \begin{cases}
			1 \mathrm{ \quad if \,} x = t, \\
			0 \mathrm{ \quad if \, } x \neq t.
		\end{cases}
	\end{equation}
	In the absence of structural information, no classical strategy can efficiently solve this type of problem, resulting in a query complexity of $\mathcal{O}(N)$ for the single-target case, where $N=2^n$. In contrast, on quantum computers, Grover's search provides a quadratic acceleration, yielding a complexity of $\mathcal{O}(\sqrt{N})$, showcasing the potential superiority of quantum computing in search problems. 
	
	In correspondence with this global search problem, a type of $k$-local search problem can be formulated. In this context, the term ``$k$-local'' describes an operator where each component acts on at most $k$ qubits. For example, in $k$-local Hamiltonian problems~\cite{Kempe2006}, the $k$-local Hamiltonian is defined as $H=\sum_\alpha H_\alpha$, with each component $H_\alpha$ acting on at most $k$ qubits. If the global Oracle in Grover's search is restricted to $k$-local operators, then for an input $x$, it can only check $k$ bits of $x$ at a time within the Oracle. An analogous example is a typical processor limited to processing a fixed number of bits, such as 32 or 64 bits, rather than handling all data simultaneously. Moreover, due to the desired naturalness of the $k$-local Oracle, it should not selectively pick $k$ bits for $\left\lceil n/k \right\rceil$ times to form the global information about $t$. Instead, the Oracle must consider all $k$-combinations of $\left\{ x_j\right\}$, obtaining the count of matches between $k$-combinations of $x$ and $t$.  Besides, to maintain consistency with the Oracle of global search when $k=n$, the frequency of matches is adopted. This paper focuses on the single-target case. Thus, this objective function is represented as  
	\begin{equation}
		f_k(x) =\frac{C_d^k}{C_n^k},
	\end{equation}
	where $d=n-d_H(x,t)$, and $d_H(x,t)$ denotes the Hamming distance between $x$ and $t$. Here, $d$ represents the number of bits that $x$ matches with $t$. 
	
	As $k$ increases to $n$, the structural information inherit in $f_k(x)$ gradually diminishes, and ultimately, when $k \to n$, $f_k(x)$ reduces to the objective function of unstructured search. However, when $k$ is held constant, the explicit structural information offered by $f_k(x)$ enables an efficient solution to the $k$-local search problem on classical computers within ${\mathcal O}(n)$ Oracle calls. Here is a straightforward approach. The objective function $f_k(x)$ outputs 0 when the number of matched bits between $x$ and $t$ is less than $k$. For random $x$, the probability of this occurring is $\sum_{d=0}^{k-1} C_n^d /2^n$, which converges to 0 as $n$ increases. Thus, a few random guesses are enough to find an $x$ such that $f_k(x) > 0$. Once such an $x$ is found, its bits can be progressively adjusted to maximize $f_k(x)$ towards 1, eventually revealing the target $t$. An example of such a classical algorithm is given in Algorithm \ref{alg_classical_sol_k_local} in the appendix. The $k$-local search problem is characterized by its simplicity, and it is precisely this simplicity that facilitates the subsequent derivation.

	\subsubsection{\textit{K}-local quantum search algorithm}\label{subsubsec_circuit_k_local_QS}
	Regardless of the workspace qubits used by the Oracle, the circuit for Grover's search is given by
	\begin{equation}
		\left| \psi \right> = {\left( H^{\otimes n} (2\left|0\right>^{\otimes n} \left<0\right|^{\otimes n} -I)H^{\otimes n}O\right)}^p \left| + \right>^{\otimes n}. 
	\end{equation}
	where $\left|0\right>$ is the zero state, and $\left<0\right|$ is its conjugate transpose. The Hadamard gate is denoted as $H$, and $\left| + \right> = H \left| 0 \right>$. As illustrated in Eq.~(\ref{eq_grover_Oracle}), the evolution of Grover Oracle $O$ can be expressed as a time evolution operator as $e^{-i\pi H_G}$, where $H_G$ is the Grover Hamiltonian defined as $H_G \left| x \right> = f(x) \left| x \right>$. Similarly, neglecting the global phase ${{e}^{-i\pi }}$, the term $2\left|0\right>^{\otimes n} \left<0\right|^{\otimes n} -I$ can be expressed as $e^{-i\pi H_{G,0}}$, where $H_{G,0}$ corresponds to $H_G$ with the target $t=0$. Accordingly, the circuit of Grover's search is reformulated as 
	\begin{equation}
		\left| \psi \right> = {\left( H^{\otimes n} e^{-i\pi H_{G,0}}H^{\otimes n} e^{-i\pi H_G} \right)}^p \left| + \right>^{\otimes n}. 
	\end{equation} 
	Grover's search involves a rotation in the target space. Thus, the factor $\pi$ can be generalized to a coefficient $\theta$, where $\theta$ typically does not exceed $\pi$.
	
	Similarly, by defining the Hamiltonian for $k$-local search as
	\begin{equation}\label{eq_H_k}
		H_k\left| x \right> = f_k(x) \left| x \right>,
	\end{equation}
	the evolution of $k$-local quantum search can be formulated as
	\begin{equation}\label{eq_k_local_QS}
		\left| \psi \right> = {\left( H^{\otimes n} e^{-i\pi H_{k,0}} H^{\otimes n} e^{-i\pi H_k} \right)}^{p} \left| + \right>^{\otimes n},
	\end{equation}
	where $H_{k,0}$ represents $H_k$ with target set to $t=0$. Here, $\pi$ can also be generalized to a parameter $\theta$, in which case the required number of iterations, $p_\theta$, should scale as $\mathcal{O}(\theta^{-1} p)$. When $k=n$, this circuit reduces to Grover's search. Our primary focus, however, lies in the cases where $k$ is a small constant. In this scenario, the circuit complexity of a single iteration is $\mathcal{O}(n^k)$. Additionally, the operator $H^{\otimes n} e^{-i\pi H_{k,0}} H^{\otimes n}$ can be expressed as $e^{-i\pi H_{B, k}}$, where $H_{B,k} = H^{\otimes n}  H_{k,0} H^{\otimes n}$. Further technical details about the circuit are provided in Appendix~\ref{apsubsec_PH_k_SAT}. Using Trotter-Suzuki formula $e^{-i\theta (H_A+H_B)}  =$ $e^{-i\theta/2 H_A}e^{-i\theta H_B}e^{-i\theta/2 H_A} + \mathcal{O}(\theta^{3})$, the entire evolution can be approximated as 
	\begin{equation}\label{eq_trotterized_k_local_QS}
		\left( e^{-i\theta H_{B,k}}e^{-i\theta H_{k}} \right)^{p_\theta} \left|+ \right> =  e^{-i\theta p_\theta \mathcal{H}_{k}}  \left| + \right> + \mathcal{O}(n^{-1}) +\mathcal{O}(\theta^{3} p_\theta), 
	\end{equation}
	where $\mathcal{H}_{k} = H_{B,k} + H_{k}$. We focus on the case of $\theta = \Theta(n^{-1})$. Based on this, we introduce two straightforward Lemmas \ref{lemma_gap_assume} and \ref{lemma_k_local_QS_assume}. A more rigorous version is presented and proven in Lemmas \ref{lemma_gap} and \ref{lemma_k_local_QS}. Here, the spectral gap refers to the energy gap between the eigenstate with the highest and second-highest energy in $\mathcal{H}_{k}$. Our further derivation is developed based on these lemmas.
	
	\begin{lem}\label{lemma_gap_assume}
		For a small constant $k$, the spectral gap $g_k$ of $\mathcal{H}_{k}$ scales as $\Theta(n^{-1})$. 
	\end{lem}
	\begin{lem}\label{lemma_k_local_QS_assume}
		For a small constant $k$ and $\theta = \Theta(n^{-1})$, ${\mathcal O}(n^2)$ iterations are necessary for the $k$-local quantum search to address the $k$-local search problem. 
	\end{lem}
	
	\subsection{Efficiency  on random max-\textit{k}-SSAT problem}\label{sec_efficiency}
	The goal of this section is to establish the efficiency of $k$-local quantum search and its adiabatic variant when applied to random $k$-SAT instances. Since a search algorithm necessitates the existence of a target, we focus on a new problem, max-$k$-SSAT. This problem is a restricted variant of max-$k$-SAT that focuses specifically on satisfiable instances, asking for an interpretation $t$ that satisfies all the clauses (and thus also the maximum possible number of clauses). The objective function for this optimization is defined as $f(x) = \sum_\alpha f_\alpha(x)$, where $\alpha$ indexes the clauses and $f_\alpha(x)$ represents their characteristic function. The following lemma demonstrates the significance of solving max-$k$-SSAT in relation to the $k$-SAT decision problem.

	\begin{lem}\label{lem_problem_reduction}
		If max-$k$-SSAT is solved within polynomial time of magnitude $\mathcal{O}(f(n))$, the $k$-SAT decision problem can similarly be solved in the same complexity.
	\end{lem}
	\begin{proof}
		Partition the complete set $U$ of $k$-SAT instances into two subsets: $U_s$, consisting of all satisfiable instances, and $U_u$, containing the unsatisfiable ones. Suppose there is an algorithm $A$ that solves max-$k$-SSAT in polynomial time ${\mathcal O}(f(n))$. Using $A$, we can construct an algorithm $A'$ to solve the $k$-SAT decision problem. Algorithm $A'$ takes any instance $I \in U$ and runs $A$ on it. If $A$ does not terminate within time ${\mathcal O}(f(n))$, $A'$ halts and outputs a random $x$. For any instance $I$, $A'$ terminates within time ${\mathcal O}(f(n))$ and produces a result $x$. If $I \in U_s$, $x$ corresponds to a satisfying interpretation $t$ of the Boolean formula. If $I \in U_u$, $x$ is random, indicating unsatisfiability. The cost of this verification process does not exceed ${\mathcal O}(f(n))$, ensuring an overall complexity of ${\mathcal O}(f(n))$.
	\end{proof}
	
	To describe the distribution of satisfiable instances, we introduce two random models. The first model, $F_s(n,m,k)$, derives from the original model $F(n,m,k)$ by selectively choosing clauses while ensuring satisfiability. The second model, $F_f(n,m,k)$, serves as an approximation of $F_s(n,m,k)$ for theoretical analysis. Specifically, in $F_f(n,m,k)$, a predetermined interpretation $t_0$ is randomly provided, and only the clauses satisfied by $t_0$ are selected. An example of this clause selecting process is presented in Figure \ref{fig_clause} in the appendix. Our primary objective is to demonstrate the efficiency of $k$-local quantum search algorithms for random max-$k$-SSAT instances generated from $F_f(n,m,k)$ with specific $m$, and further extend the conclusion to $F_s(n,m,k)$. 
	
	\subsubsection{Max-\textit{k}-SSAT: \textit{k}-local search with missing clauses}\label{subsec_k_SAT_k_search} 
	
	For max-$k$-SSAT instances, given the objective function $f(x) = \sum_\alpha f_\alpha(x)$, the contribution from each clause to the problem Hamiltonian is defined as $H_\alpha \left| x \right> = f_\alpha \left| x \right>$. Thus, its corresponding problem Hamiltonian is $H_C  = \sum_\alpha H_\alpha$. The detail about the construction of $H_C$ is presented in  Appendix~\ref{apsubsec_PH_k_SAT}. Given that the clause $\alpha$ is randomly selected under the condition that it is satisfiable by the predetermined $t$, it is natural to conceptualize the diagonal elements (eigenvalues) of Hamiltonian $H_\alpha$ as random variables, where the eigenvalue corresponding to $\left| x \right>$ is denoted as $E_{\alpha, x}$. The mean and variance of $E_{\alpha, x}$ can be derived through statistical analysis of the clause set.
	
	Let the set of clauses satisfiable by $t$ be denoted as $S_t$. $S_t$ is divided into $k$ subsets $S_{t,j}$, where $1 \le j \le k$ represents the number of literals in the clause satisfied by $t$. Within each subset, there are $C_k^j C_n^k$ clauses, and the count of unsatisfied clauses of $x$ is $C_d^{k-j} C_{n-d}^j$, where $d=n-{{d}_{H}}(x,t)$. Following this analysis, the total number of clauses satisfied by $x$ is $\sum_{j=1}^k (C_k^j C_n^k – C_d^{k-j} C_{n-d}^j)$, thus, $(2^k-2)C_n^k + C_d^k$. Given that when $x$ is satisfied by $\alpha$, $E_{\alpha, x } = 1$, and otherwise, $E_{\alpha, x } = 0$, the mean and variance of $E_{\alpha, x}$ are
	\begin{equation}
		{{\mu }_{k,x}}=\frac{2^k-2}{2^k-1}+\frac{C_d^k}{\left( 2^k-1 \right)C_n^k}\le 1,\;
		\sigma _{k,x}^2=\left( 1-\mu _{k,x} \right)^2 \mu _{k,x}+\frac{\mu _{k,x}^2\left( C_n^k-C_d^k \right)}{\left( {{2}^{k}}-1 \right)C_{n}^{k}}\le \frac{1}{{{2}^{k}}-1}.
	\end{equation}
	
	The problem Hamiltonian $H_C$ for a random $k$-SAT instance is the sum of a series of random $H_\alpha$ terms. Thus, its eigenvalues naturally constitute random variables, denoted as ${\mathcal{E}}_{k,x} = \sum_\alpha E_{\alpha, x}$. Since each $\alpha$ is randomly selected from the same set with replacement, these $E_{\alpha, x}$ can be treated as independent and identically distributed (i.i.d.) random variables. For a given $x$, according to the central limit theorem, ${\mathcal{E}}_{k,x}/m$ approximately follows a normal distribution such that
	\begin{equation}\label{eq_normal_distribution}
		\sqrt{m}\left( \frac{1}{m}{\mathcal{E}}_{k,x}-{{\mu }_{k,x}} \right)\sim N(0,\sigma _{k,x}^{2}).
	\end{equation}
	This phenomenon is particularly intriguing: for the problem Hamiltonian $H_C$ of a random instance in $F_f(n,m,k)$, $H_C/m$ converges in probability to a certain ``standard form'' as $m$ increases. By disregarding the global phase and normalizing the eigenvalue range to $\left[0, 1\right]$, namely, letting the normalized problem Hamiltonian
	\begin{equation}\label{eq_bar_HC}
		\bar{H}_C = \frac{(2^k-1)H_C}{m} - (2^k-2),
	\end{equation}
	the expectation of $\bar{H}_C \left| x\right>$ is $\frac{C_{d}^{k}}{C_{n}^{k}}\left| x\right>$. Notably, the problem Hamiltonian $H_k$ of $k$-local search, shown in Eq.~(\ref{eq_H_k}), takes the same form. In other words, $H_k$ represents the expected value of $\bar{H}_C$.
	
	Indeed, these standard Hamiltonians for random $k$-SAT in $F_f(n,m,k)$ suggest a scenario where all clauses are equally likely to be selected, naturally leading to the same form as $H_k$. Since $k$-local search problems represents instances where all possible clauses are selected, random instances of $k$-SAT with interpretations can thus be viewed as $k$-local search problems where some clauses are missing. This raises a \textit{question}: during the process of omitting clauses, at what point does the average-case complexity of these random instances transition from polynomial to exponential? In the rest of this section, we explore this question by providing an upper bound.
	
	\subsubsection{Efficiency of \textit{k}-local quantum search}\label{subsec_threshold_QS}
	
	By introducing the normalized problem Hamiltonian $\bar{H}_C$ of $k$-SAT instances to replace the problem Hamiltonian of the $k$-local searches, the circuit for dealing $k$-SAT is reformulated as
	\begin{equation}
		\left| \psi \right> = {\left( e^{-i\theta H_{B, k}} e^{-i\theta\bar{H}_C} \right)}^{p_\theta} \left| + \right>^{\otimes n}.
	\end{equation}
	As analyzed in Section \ref{subsec_k_SAT_k_search}, the normalized problem Hamiltonian $\bar{H}_C$ for random instances in $F_f(n,m,k)$ exhibits significant convergence to $H_k$, meaning that the deviation $\Delta H_C = \bar{H}_C - H_k$ asymptotically converges to 0 in probability as $m$ increases. Specifically, given the normal distribution shown in Eq.~(\ref{eq_normal_distribution}), the eigenvalue ${\mathcal{E}}_{k,x}$ of $H_C$ satisfies
	\begin{equation}\label{eq_Ekx_range}
		\left | \frac{1}{m}{{\mathcal{E}}_{k,x}} - {{\mu }_{k,x}} \right| \le \frac{c}{\sqrt{m({{2}^{k}}-1)}}
	\end{equation}
	with a probability of $\mathrm{erf}(c/\sqrt{2})$, where $\mathrm{erf}(x)$ denotes the error function.
	
	\begin{thm}\label{theo_threshold}
		For any small $\epsilon>0$ and sufficiently large $n$, the $k$-local quantum search algorithm, when applied to random instances in $F_f(n,m,k)$ with $m=\Theta (n^{2+\epsilon+\delta})$, achieves efficiency comparable to that of the $k$-local search problem, with a probability of $1-\mathcal{O}(\mathrm{erfc}(n^{\delta/2}))$.
	\end{thm}
	\begin{proof}\label{proof_threshold}
		For simplicity, we focus on the case where $\theta = \pi$ and $p = \mathcal{O}(n)$. The proof for other $\theta < \pi$ follows analogously, as the underlying logic and structure remain unchanged. By introducing $c=n^{\delta/2}$ and normalizing the eigenvalues, Eq.~(\ref{eq_Ekx_range}) is reformulated as 
		\begin{equation}\label{eq_Ekx_devia}
			\left |{{\mathcal{\bar{E}}}_{k,x} -E_{k,x}} \right| \le \sqrt{\frac{{{2}^{k}}-1}{m}}n^{\delta/2}
		\end{equation}
		with a probability of $1-\mathcal{O}(\mathrm{erfc}(n^{\delta/2}))$, where $\mathcal{\bar{E}}_{k,x}$ is the eigenvalue of $\bar{H}_C$ as presented in Eq.~(\ref{eq_bar_HC}), and $E_{k,x}$ is the eigenvalue of $H_k$. Given the magnitude of $m$ as $\Theta (n^{2+\epsilon+\delta})$, the deviation of $\mathcal{\bar{E}}_{k,x}$ from $E_{k,x}$ is probabilistically bounded by $\mathcal{O}(n^{-(1+\epsilon/2)})$, which is asymptotically $o(n^{-1})$.
		
		The evolution of $k$-local quantum search can be expressed as 
		\begin{equation*}
			\left| \psi \right> = {\left( E_B(E_C+\Delta E_C)  \right)}^{p} \left| + \right>^{\otimes n},
		\end{equation*}
		where $E_B= e^{-i\pi H_{B,k}} $, $E_C=e^{-i\pi H_{k}}$, and $\Delta E_C = e^{-i\pi \bar{H}_C} - e^{-i\pi H_{k}}$. When $\Delta E_C$ acts on the quantum state, each eigenvalue of $\Delta H_C$ independently affects the computational basis states. This influence can be mathematically represented as
		\begin{equation*}
			(E_C + \Delta E_C) \sum_x \alpha_x \left| x \right> = \sum_x \alpha_x e^{-i\pi (E_{k,x}+\Delta\bar{\mathcal{E}}_{k,x}) } \left| x \right>. 
		\end{equation*}
		 If this quantum state is measured, it collapses to a particular computational basis state $\left| x \right>$. In this context, we denote $\left|\Delta E_C \right| = o_r(n^{-1})$ to indicate that the deviation caused by $\Delta E_C$ is $o(n^{-1})$ with a probability of $1-\mathcal{O}(r)$, where the error probability $r = \mathrm{erfc}(n^{\delta/2})$, as shown in Eq.~(\ref{eq_Ekx_devia}). 

		Since the cumulative effect of the high-order terms regrading $\Delta E_C$ can exceed ${o}(1)$ only in a probability significantly smaller than $\mathcal{O}(r)$, the operator $U$ for the entire evolution is expanded as
		\begin{equation*}
			U=(E_B E_C)^{p} + \sum_{j=0}^{p-1}[(E_BE_C)^j E_B\Delta E_C(E_BE_C)^{p-j-1}]  + {o_r}(1),
		\end{equation*}
		resulting in the evolved state expressed as
		\begin{equation*}
			\left| \psi' \right> = \left| \psi_0 \right> +  \sum_{j=0}^{p-1}[(E_BE_C)^j E_B\Delta E_C(E_BE_C)^{p-j-1}]\left| + \right>^{\otimes n}   +  {o_r}(1),
		\end{equation*}
		where $\left| \psi_0 \right>$ is the resulting state for $k$-local search problem. Focusing on the amplitude of the target state $\left| t \right>$, the deviation from $\left< t | \psi_0 \right>$ is given by
		\begin{equation*}
			\Delta P = \sum_{j=0}^{p-1} \left< t \right| (E_BE_C)^j E_B\Delta E_C(E_BE_C)^{p-j-1} \left| + \right>^{\otimes n}   +  {o_r}(1).
		\end{equation*}
		
		Denoting each term in $\Delta P$ as $\Delta P_j$, every $\Delta P_j$ is of order $o_r(n^{-1})$, and their summation should remain $o_r(1)$. Specifically, if we regard each $\Delta P_j$ as a random variable, the error rate can be maintained within $\mathcal{O}(\mathrm{erfc}(n^{\delta/2}))$ when these variables are dependent. In the case of independence, the error rate can be further reduced due to the central limit theorem. Consequently, the overall deviation is $o(1)$ with a probability of $1-\mathcal{O}(\mathrm{erfc}(n^{\delta/2}))$. In other words, after the evolution, the amplitude of the target computational basis state remains on the same order of magnitude as that of the $k$-local search problem, with a probability of $1-\mathcal{O}(\mathrm{erfc}(n^{\delta/2}))$. 
	\end{proof}
	
	\subsubsection{Efficiency of \textit{k}-local adiabatic quantum search}\label{subsec_threshold_AQS}
	
	For max-$k$-SSAT instances with $m$ significantly smaller than $\Theta(n^2)$, the deviation $\Delta H_C = \bar{H}_C - H_k$ becomes so substantial, that the efficiency of $k$-local quantum search cannot be maintained. To address this, we utilize quantum adiabatic computation~\cite{Kadowaki1998, farhi2001QAA, Albash2018AQC} to mitigate the impact from this deviation. The adiabatic theorem states that a physical system remains in its instantaneous eigenstate if a given perturbation is acting on it slowly enough and if there is a gap between the eigenvalue and the rest of the Hamiltonian's spectrum~\cite{Born1928}.
	
	Since the $k$-local search represents the average case of max-$k$-SSAT, we define the system Hamiltonian for $k$-local search as
	\begin{equation}\label{eq_Hamiltonian_AQS}
		\bar{H}_k(s) = sH_k + (1-s)H_{B,k}. 
	\end{equation}
	Notably, the initial state $\left| + \right>^{\otimes n}$ is the most exited state of $H_{B, k}$, corresponding to the highest eigenvalue of the system's initial Hamiltonian. By increasing $s$ from 0 to 1 slowly enough, the system remains in the most excited state of $\bar{H}_k(s)$ throughout the evolution. When $s=1$, the system reaches the desired target state $\left| t \right>$, which maximizes the objective function $f_k(x)$.
	
	In this paper, we adopt a linearly varying $s = t/T$, analogous to the quantum adiabatic algorithm (QAA)~\cite{farhi2000QAA}. The system Hamiltonian for $k$-local search is expressed as 
	\begin{equation}
		\bar{H}_k(t) = \frac{t}{T}H_k+ (1-\frac{t}{T})H_{B,k}.
	\end{equation}
	The efficiency of this process depends on the total evolution time $T$, which is determined by the minimal gap $g_0$ of the system Hamiltonian \cite{farhi2000QAA}. In this context, the spectral gap of Hamiltonian $\bar{H}_k(s)$ represents the energy difference between the largest and second-largest eigenvalues, denoted as $g_k(s)$. The minimal gap $g_{k,0} = \min_s g_k(s)$ occurs at $s = 1/2$ that scales as $\Theta(n^{-1})$, as illustrated in Lemma \ref{lemma_gap_assume}. Since $T$ must satisfy $T \gg g^{-2}_0$, the evolution time is of order $\mathcal{O}(n^2)$. More details about quantum adiabatic computation is presented in Appendix \ref{apsec_prelim}. 
	
	We denote the system Hamiltonian when applying this algorithm to $k$-SAT instances as  $H_k(s)$ or $H_k(t)$. Under this scenario, $H_k$ is replaced by $\bar{H}_C$, which introduces a deviation $\Delta H_C = \bar{H}_C - H_k$. Quantum adiabatic evolution is inherently robust to such deviation. During the evolution, $k$-local quantum search needs to balance the number of iterations with the minimal gap, while $k$-local adiabatic quantum search focuses solely on the minimal gap. Specifically, when applying $k$-local quantum search, if the evolution time $T$ exceeds $\mathcal{O}(n)$,  the accumulative evolution of deviation $\Delta H_C$ becomes so pronounced that it leads to a reduction in performance. However, in the adiabatic variant, a longer evolution time tends to improve performance, despite increasing time complexity. This trade-off between evolution time and performance is formalized in Lemma \ref{lemma_AQS_on_SAT}, which naturally leads to Theorem \ref{theo_adiabatic_k_local_QS}.
	
	\begin{lem}\label{lemma_AQS_on_SAT}
		Given the efficiency of $k$-local quantum search on random max-$k$-SSAT instances in $F_f(n,m,k)$ for $m = \Theta(f(n))$ with $\theta = \Theta(n^{-1})$ and $p_\theta = \mathcal{O}(n^2)$, the $k$-local adiabatic quantum search achieves similar efficiency for $m = \Theta\left(f^{1/2}(n)\right)$ with an evolution time of $\mathcal{O}(n^2)$.
	\end{lem}
	\begin{proof}
		By incorporating $\theta = \Theta(n^{-1})$, $p_\theta = \mathcal{O}(n^2)$ and $\bar{H}_C$ into Eq.~(\ref{eq_trotterized_k_local_QS}), the evolution of $k$-local quantum search can be approximated as
		\begin{equation*}
			\left( e^{-i\theta H_{B,k}}e^{-i\theta \bar{H}_C} \right)^{p_\theta} \left|+ \right> =  e^{-i\theta p_\theta \mathcal{H}_{C, k}}  \left| + \right> + \mathcal{O}(n^{-1}),
		\end{equation*}
		where $\mathcal{H}_{C, k} = H_{B, k} + \bar{H}_C$. This can be interpreted as rotations in the eigenspace of $\mathcal{H}_{C, k}$. If the $k$-local quantum search is generally efficient on $F_f(n,f(n),k)$, it implies that during the evolution of $\mathcal{H}_{C,k}$, the cumulative deviation from $\Delta H_C$ does not significantly affect the spectral gap $g'_{k}$ of $\mathcal{H}_{C, k}$ within $\mathcal{O}(n)$ time. Lemma~\ref{lemma_gap_assume} establishes that the gap $g_k$ of $\mathcal{H}_k$ scales as $\Theta(n^{-1})$. Denoting the deviation in the gap as $\Delta g_k = g'_{k} - g_k$, the cumulative deviation in phase difference $\theta p_\theta \Delta g_k$ between the eigenstates should remain on the order of $o(n^{-1})$ to maintain the general efficiency, implying $\Delta g_k$ must scales as $o(n^{-2})$. Thus, when applying the $k$-local adiabatic quantum search, smaller values of $m$ become feasible.

		Since the gap $g_k(s)$ of $\bar{H}_k(s)$ reaches a minimum of order $\Theta(n^{-1})$ at $s=1/2$, the system Hamiltonian $H_k(s)$ for $\bar{H}_C$ should similarly achieves its minimal gap $g'_{k,0}$ around $s=1/2$. This is because $\Delta H_C$ remains infinitesimal compared to $H_k$ when $m = \Omega(n)$. Thus, our analysis focuses on the approximate minimal gap $g'_k$ at $s=1/2$, where $g'_k$ is expected to remain in the same order as $g'_{k,0}$. For $m=f(n)$, the deviation $\Delta H_C$ from $H_k$ is generally of the order ${\mathcal O}(f^{-1/2}(n))$, leading to a gap deviation $\Delta g'_k = o (n^{-2})$ from $g_k$. When $m$ is reduced to $f^{1/2}(n)$, $\Delta H_C$ grows to a magnitude of ${\mathcal O}(f^{-1/4}(n))$, naturally increasing the gap deviation $\Delta g'_k$ to the order of $o (n^{-1})$. Nonetheless, this deviation does not alter the overall order of $g'_k$. Therefore, the efficiency of $k$-local adiabatic quantum search on random max-$k$-SSAT instances in $F_f(n,m,k)$ is preserved with $m = \Theta(f^{1/2}(n))$.
	\end{proof}
	
	\begin{thm}\label{theo_adiabatic_k_local_QS}
		For any small $\epsilon>0$, sufficiently large $n$ and random max-$k$-SSAT instances in $F_f(n,m,k)$ with $m=\Theta (n^{1+\epsilon+\delta})$, the $k$-local adiabatic quantum search achieves performance comparable to its application to the $k$-local search, with a probability of $1-\mathcal{O}(\mathrm{erfc}(n^{\delta/2}))$.
	\end{thm}
	\begin{proof}
		According to Theorem \ref{theo_threshold}, the $k$-local quantum search is generally efficient on random instances of $k$-SAT in $F_f(n,m,k)$ with $m=\Theta (n^{2+\epsilon+\delta})$ with a probability of $1-\mathcal{O}(\mathrm{erfc}(n^{\delta/2}))$. Here, the term $2+\epsilon$ ensures the circuit's efficiency, while the additional $\delta$ controls the error probability to $\mathcal{O}(\mathrm{erfc}(n^{\delta/2}))$.  Based on Lemma \ref{lemma_AQS_on_SAT}, for instances with $m = \mathcal{O}(n^{1+\epsilon/2})$, the efficiency is maintained, with a probability on the order of $1-\mathcal{O}(\mathrm{erfc}(n^{\delta/2}))$. Finally, by substituting $\epsilon/2$ with $\epsilon$, the theorem is established.
	\end{proof}
	
	\subsubsection{Proof of main theorem}\label{sec_main_theo}
	
	In the previous analysis, we have established efficiency of these algorithms on random max-$k$-SSAT under distribution $F_f(n,m,k)$. To build up the connection with $F_s(n,m,k)$, we begin by demonstrating a key reduction in Lemma~\ref{lem_model_reduction}. An example of the clause selecting process in the two random models is illustrated in Figure \ref{fig_clause} in the appendix. 
	\begin{lem}\label{lem_model_reduction}
		For any small $\epsilon>0$, sufficiently large $n$ and  $m=\Omega(n^{1+\epsilon})$, if random max-$k$-SSAT under distribution $F_f(n,m,k)$ is solved polynomially on average, random max-$k$-SSAT under $F_s(n,m,k)$ can also be solved polynomially on average.
	\end{lem}
	\begin{proof}
		By arranging clauses into a sequence that allows for repetition and is order-dependent, each specific instance can be characterized by this sequence. Accordingly, for any given satisfiable instance identified by a sequence, it is impossible to determine which model was used to generate it. That is, the random models $F_s(n,m,k)$ and $F_f(n,m,k)$ generate identical set of instances. The only distinction between these models lies in the probability distribution over these instances.
		
		$F_s(n,m,k)$ is a natural derivation from the standard $F(n,m,k)$ by directly excluding these ``unsatisfiable branches''. Thus, each instance is generated with uniform probability. On the other hand, $F_f(n,m,k)$ pre-assigns a randomly chosen target $t_0$. This results in $N$ potential ``entries'' for generating instances, where each entry produces instances with equal probability. When considering a specific instance, the number of possible entries is restricted by its number of interpretations. If an instance has only one interpretation, it corresponds to a single entry, giving it the same probability as other instances with a single interpretation. However, for an instance with $q$ interpretations, it has $q$ times the probability due to the multiple corresponding entries.
		
		In the random model $F_f(n,m,k)$, instances with exponentially many interpretations have significantly higher individual probabilities compared to those with polynomially many interpretations. However, as $m=\Omega(n^{1+\epsilon})$, the proportion of such instances among all possible instances decreases super-exponentially. Since these instances are inherently simple, they do not impact the magnitude of the overall average-case complexity. Instances with polynomially many interpretations have individual probabilities that differ only by polynomial factors, facilitating the reduction between random problems. Detailed description of reduction between random problem is presented in Appendix \ref{apsubsec_ac_comp}.
	\end{proof}

	To establish the average-case complexity of max-$k$-SSAT, another difficulty is the computational complexity of instances where these local quantum search algorithm becomes inefficient. Given the input of max-$k$-SSAT is restricted to satisfiable instances, Grover's search can be effectively applied to handle these inefficient instances. An example of such a algorithm is provided in Algorithm \ref{alg_main} in the appendix. Based on this framework, a threshold where the random max-$k$-SSAT becomes polynomial on average can be established as follows.

	\begin{proof}[Proof of Theorem \ref{theo_main}]
		Given the reduction presented in Lemma \ref{lem_model_reduction}, our analysis focuses specifically on the average-case complexity of max-$k$-SSAT with the random model $F_f(n,m,k)$. Given $m=\Omega(n^{2+\epsilon})$, we expresses the coefficient $2+\epsilon$ as $1+\epsilon'+\delta'$, where $\epsilon' = \epsilon/2$ and $\delta' = 1+\epsilon/2$. According to Theorem \ref{theo_adiabatic_k_local_QS}, by substituting the $\epsilon'$ and $\delta'$ for $\epsilon$ and $\delta$, respectively, the vast majority of random instances can be effectively addressed by $k$-local adiabatic quantum search with a probability of $1-\mathcal{O}(\text{erfc}(n^{1/2+\epsilon/4}))$ within $\mathcal{O}(n^2)$ time. Consequently, the average-case complexity of these instances is decisively polynomial with respect to $n$.
		
		Regarding the remaining instances, although Algorithm \ref{alg_main} gradually increases the evolution time to handle them, we consider the worst-case scenario in which all of these instances require $\Theta(\sqrt{N})$ time. This leads to an average-case time complexity of $\mathcal{O}(\text{erfc}(n^{1/2+\epsilon/4})2^{n/2})$. Given that as $x \to \infty$, $\text{erfc}(x) \sim {e^{-x^2}}/{x \sqrt{\pi}}$, for a sufficiently larger $n$, the contribution to the average-case complexity from these extreme instances is on the order of $\mathcal{O}(2^{n/2}e^{-n^{1+\epsilon/2}})$.  Thus, the overall average-case complexity remains polynomial.
	\end{proof}
	
	
	\section{Landscape of average-case complexity for random max-\textit{k}-SSAT}\label{sec_complexity}

	This section presents the landscape of average-case complexity for random max-$k$-SSAT, focusing on the variation of $m$ around and beyond $r_kn$. We establishes that max-$k$-SSAT reaches its peak average-case complexity when $m$ lies within the interval $\left[ (r_k-\epsilon)n, (r_k+\epsilon)n \right]$, for any small $\epsilon>0$ and sufficiently large $n$. Beyond this range, the average-case complexity decreases, due to the increasing efficiency of $k$-local adiabatic quantum search. Finally, when $m$ exceeds $\Theta(n^2)$, the computational complexity of max-$k$-SSAT becomes polynomial on average. 
	
	When $m$ is small, $k$-SAT remains inherently easy due to the abundance of exponential many satisfying interpretations. However, as $m$ increases and approaches the threshold $r_kn$, $k$-SAT undergoes a transition from solvability to insolvability. In the context of max-$k$-SSAT, this transition marks the problem's most challenging scenario. Lemmas \ref{lemma_lower_bound} and \ref{lemma_upper_bound} respectively establish the lower and upper bound for this challenging regime. Additionally, Proposition \ref{prop_hardest} precisely identifies this critical range within $\left[ (r_k-\epsilon)n, (r_k+\epsilon)n \right]$, for any small $\epsilon > 0$ and sufficiently large $n$. 
		
	\begin{lem}\label{lemma_lower_bound}
		Let $m_l$ be such that for $\epsilon > 0$, 
		\begin{equation} \label{eq_lemma_lower_bound}
			\lim_{n \to \infty} {\mathrm{Pr}[F(n,m,k) \mathrm{\,is\,satisfiable }]=
				\begin{cases}
					1 \mathrm{ \quad if \,} m = m_l, \\
					0 \mathrm{ \quad if \, } m = (1+\epsilon)m_l,
			\end{cases}}
		\end{equation}
		then for any $m_0<m_l$, the average-case complexity of $F_s(n,m_0,k)$ is no greater than $F_s(n,m_l,k)$. 
	\end{lem}
	\begin{proof}
		Assume there exists an algorithm $A_1$ that can find satisfying interpretations for instances in $F_s(n,m_l,k)$ with an average-case complexity of ${\mathcal O}(g_1(n))$.  Let $I_\alpha$ be a random instance in $F_s(n,m_0,k)$. By adding $(m_l - m_0)$ clauses to $I_\alpha$, a new instance $I'_\alpha$ with $m = m_l$ is formed. Because the majority of instances in $F(n,m_l,k)$ are satisfiable, only a constant number of attempts on average are required to ensure that $I'_\alpha$ belongs to $F_s(n,m_l,k)$. Therefore, utilizing the algorithm $A_1$, the average-case complexity of solving instances in $F_s(n,m_0,k)$ remains ${\mathcal O}(g_1(n))$.
	\end{proof}
	
	\begin{lem}\label{lemma_upper_bound}
		Let $m_u$ be such that for instances in $F_s(n,m_u,k)$, the number of satisfying assignments is typically ${\mathcal O}(1)$, i.e., 
		\begin{equation*}
			\lim_{n \to \infty} {\mathrm{Pr}[\mathrm{Count}(\mathrm{Int}(F_s(n,m,k))) = \mathcal{O}(1)] =
				\begin{cases}
					0 \mathrm{ \quad if \,} m = (1-\epsilon)m_u, \\
					1 \mathrm{ \quad if \, } m = m_u,
			\end{cases}}
		\end{equation*}
		then for any $m_0>m_u$, the average-case complexity of $F_s(n,m_0,k)$ is no greater than $F_s(n,m_u,k)$. 
	\end{lem}
	\begin{proof} 
		Assume there exists an algorithm $A_2$ capable of finding all satisfying interpretations for random instances in $F_s(n,m_u,k)$ with an average-case complexity of ${\mathcal O}(g_2(n))$. For any instance $I_{\beta} \in F_s(n, m_0, k)$, it can be transformed into a new instance $I'_{\beta}$ by randomly and uniformly removing $(m_0-m_u)$ clauses.  This transformation effectively turns $I'_{\beta}$ into a random instance of $F_s(n,m_u,k)$, which can be solved with an average-case complexity of ${\mathcal O}(g_2(n))$. With only ${\mathcal O}(1)$ additional checks on average, the original instance $I_{\beta}$ can also be solved effectively.
	\end{proof}
	
	\begin{prop}\label{prop_hardest} 
		For max-$k$-SSAT under the random model $F_s(n,m,k)$, given any small $\epsilon > 0$ and sufficiently large $n$, random instances with $m \in \left[ (r_k-\epsilon)n, (r_k+\epsilon)n \right]$ exhibit average-case complexity that is not lower than instances with other values of $m$. 
	\end{prop}
	\begin{proof}
		By Lemmas \ref{lemma_lower_bound} and \ref{lemma_upper_bound}, it follows that $m_u$ cannot be less than $m_l$. This can be proven by contradiction: suppose $m_u < m_l$. Since instances in $F_s(n,m_u,k)$ typically have a constant number of satisfying assignments, $m_u$ would satisfy the upper bound condition as stated in Eq.~(\ref{eq_lemma_lower_bound}). If $m_u < m_l$, the lower bound condition would also hold, implying that there is a threshold smaller than $m_l$ where the problem becomes satisfiable. This leads to a contradiction, as $m_l$ already approximates $r_k n$, with a slight reduction ensuring near-satisfiability for $F(n,m_l,k)$.
		
		Considering the convergence described in Eq.~(\ref{eq_theoshold_r_k}), for any small $\epsilon_1>0$ and sufficiently large $n$, instances from $F_s(n, (r_k-\epsilon_1)n,k)$ should be satisfiable for the majority of instances. Regarding $m_u$, for instances from $F_s(n,r_kn,k)$, the average number of satisfying interpretations should be polynomial; otherwise, the instances would become trivially solvable. By adding extra $\epsilon_2 n$ clauses, the number of satisfying interpretations can be reduced to ${\mathcal O(1)}$ as long as $n$ is sufficiently large. Taking $\epsilon = \min\{\epsilon_1, \epsilon_2\}$, the proposition is established. 
	\end{proof}
	
	With further increases in $m$, duo to the convergence of the normalized problem Hamiltonian to $H_k$, the performance of $k$-local adiabatic quantum search on max-$k$-SSAT improves. Specifically, there may exist an efficiency threshold $s_k \ge r_k$, at which $k$-local adiabatic quantum search becomes generally efficient for instances with $m = (s_k+\epsilon)n$. In other words, as $n$ increases, the error probability tends to zero. This raise an open problem:
	\begin{center}
		\it Whether the satisfiable threshold $r_k$ coincides with the efficiency threshold $s_k$.
	\end{center}
	This question seeks to determine whether there exist a critical range $\left[r_kn, s_kn \right]$ where $k$-local adiabatic quantum search encounters significant challenges. Once $m$ reaches the magnitude of $\Theta(n^2)$, the random instances start to exhibit characteristics similar to the $k$-local search problem, significantly enhancing the efficiency of $k$-local quantum search algorithms. Moreover, the error probability of $k$-local adiabatic quantum search decreases to such an extent that the complexity of random max-$k$-SSAT becomes polynomial on average. 
	
	
	\section{Performance}\label{sec_performance}
	\subsection{Performance on random max-3-SSAT}
	
	In quantum algorithms, the final state of the system must be measured, causing the state to collapse into a specific basis state of the measurement. In our algorithm, the resulting state can be expressed as $\left| \psi \right> = \sum_x \alpha_x \left| x \right>.$ Upon measurement, the state collapses to a particular computational basis state $\left| x \right>$ with probability $p_x = \left| \alpha^2_x\right|$. The probability of measuring the target state $\left| t \right>$, denoted as  $p_t$, is referred to as the success probability, which indicates the algorithm's performance.
	
	According to Theorem~\ref{theo_threshold}, the $k$-local quantum search, as represented in Eq.~(\ref{eq_k_local_QS}), naturally applies to random instances of max-$k$-SSAT with $m=\Omega(n^{2+\epsilon})$. To evaluate the effectiveness of this algorithm, numerical simulations are conducted on random instances in $F_s(n,m,3)$ with $m$ set to $n^2$, $2n^2$ and $4n^2$, as depicted in Figure \ref{fig_vexp1}. Here, the parameter $\theta$ is fixed as $\pi$, and the number of iterations $p_3(n)$ grows linearly with $n$. Specifically, $p_3(n)$ is determined by additional simulation on 3-local search problem, as outlined in Sectioin \ref{subsec_per_QSk}. The results exhibit a high success probability when applying 3-local quantum search on random instances in $F_s(n,m,3)$ with $m=\Theta(n^2)$, and the performance improves as $m$ increases.
	
	\begin{figure}[!t]
		\centering
		{\includegraphics[width=0.68\linewidth]{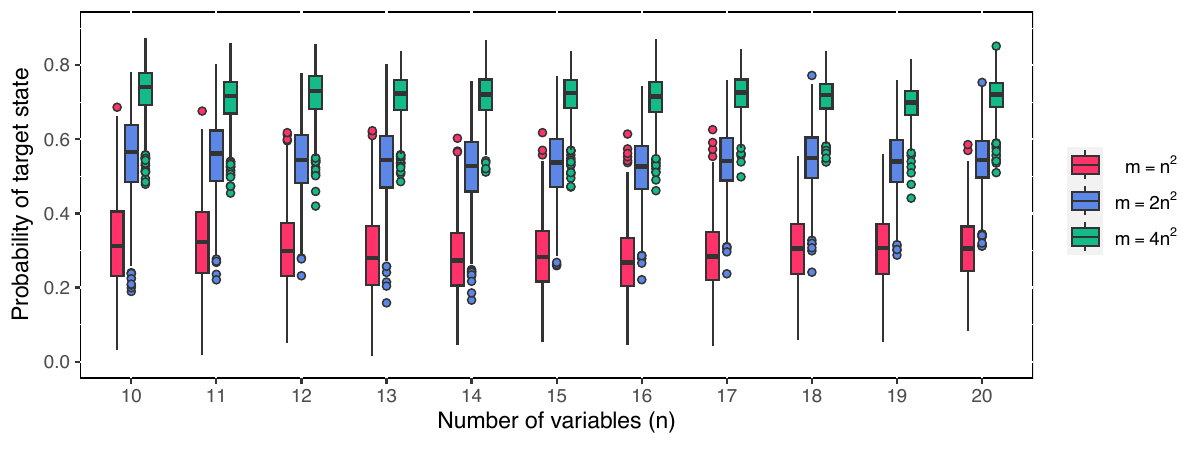}}
		\caption{\small The distribution of success probability $p_t$ for 3-local quantum search on 100 random instances in $F_s(n, m, 3)$. The number of variables $n$ ranges from 10 to 20, and the number of clauses $m=n^2$, $2n^2$ and $4n^2$, represented by red, blue and green boxes, respectively.}
		\label{fig_vexp1}
	\end{figure}
	
	Theorem~\ref{theo_adiabatic_k_local_QS} further demonstrates that the adiabatic 3-local quantum search is generally efficient for random max-3-SSAT instances with $m = \Omega(n^{1+\epsilon})$. To validate this efficiency, numerical simulations are performed on 100 random instances, as illustrated in Figure~\ref{fig_vexp2a}. In these simulations, the Trotterized circuit of the 3-local adiabatic quantum search is applied, expressed as
	\begin{equation}\label{eq_adiabatic_k_local_QS}
		\left| \psi \right> =\prod_{l=1}^{p} \left(e^{-i{\frac{p-l}{p+1}}\pi H_{B,k}}e^{-i{\frac{l}{p+1}}\pi  \bar{H}_C}\right){{\left| + \right>}^{\otimes n}},
	\end{equation}
	where $p$ is on the same order as the evolution time $T$. Here, $p$ is set to $T_3(n)$, which is of order $\mathcal{O}(n^2)$. The exact value of $T_3(n)$ is determined through additional simulations on the $k$-local search problem, as presented in Section \ref{subsec_per_QSk}. 
	
	The circuit exhibits commendable performance on random instances in $F_s(n, cn, 3)$, where $c$ ranges from 2.5 to 10. For small values of $c$, the strong performance can be attributed to the abundance of the exponential number of interpretations. Notably, the bound $r_k$, as given in Eq.~(\ref{eq_bound_rk}), is approximately 4.6986 for $k=3$. While the values $n=16, 18, 20$ are not sufficiently large to fully capture the scaling effects, a clear phase transition in solvability is observed when $c$ falls within the range $\left[4, 5.5\right]$. In this range, the algorithm occasionally encounters difficult instances, leading to a noticeable decline in performance. However, once $c$ exceeds 5.5, the likelihood of facing such challenging instances significantly decreases, and the success probability rises to a high level, consistent with the analysis in Proposition~\ref{prop_hardest} and Theorem~\ref{theo_adiabatic_k_local_QS}.
	
	\begin{figure}[!t]
		\centering
		{\includegraphics[width=0.9\linewidth]{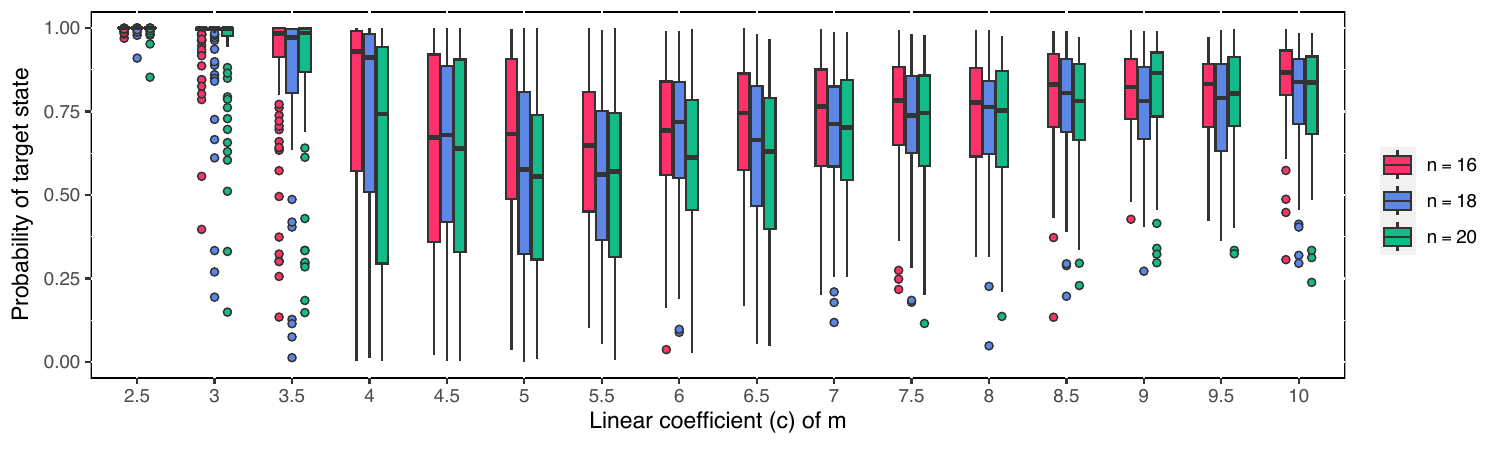}}
		\caption{\small The distribution of success probability $p_t$ for the adiabatic 3-local quantum search on 100 random instances in $F_s(n,cn,3)$. The clause density coefficient $c$ ranges from 2.5 to 10. The number of variables $n$ is 16, 18, and 20, represented by red, blue, and green boxes, respectively.}
		\label{fig_vexp2a}
	\end{figure}
	
	\subsection{Performance on \textit{k}-local search problem}\label{subsec_per_QSk}
	
	In this section, we present simulations on $k$-local search problems. For $k$-local quantum search depicted in Eq.~(\ref{eq_k_local_QS}), we illustrate the evolution of the success probability as a function of the number of iterations $p$, with parameters $n=20$, $k=1,2,3$ and $\theta = \pi$, as shown in Figure \ref{fig_k_local_search}. Besides, we report the required number of iterations to reach the first local maximum of the success probability, along with the corresponding probability values, summarized in Table~\ref{tab_QS}. The simulation results align with the theoretical predictions, confirming that $p$ should remain on the order of $\mathcal{O}(n)$.
	
	\begin{figure}
		\centering
		{\includegraphics[width=0.6\linewidth]{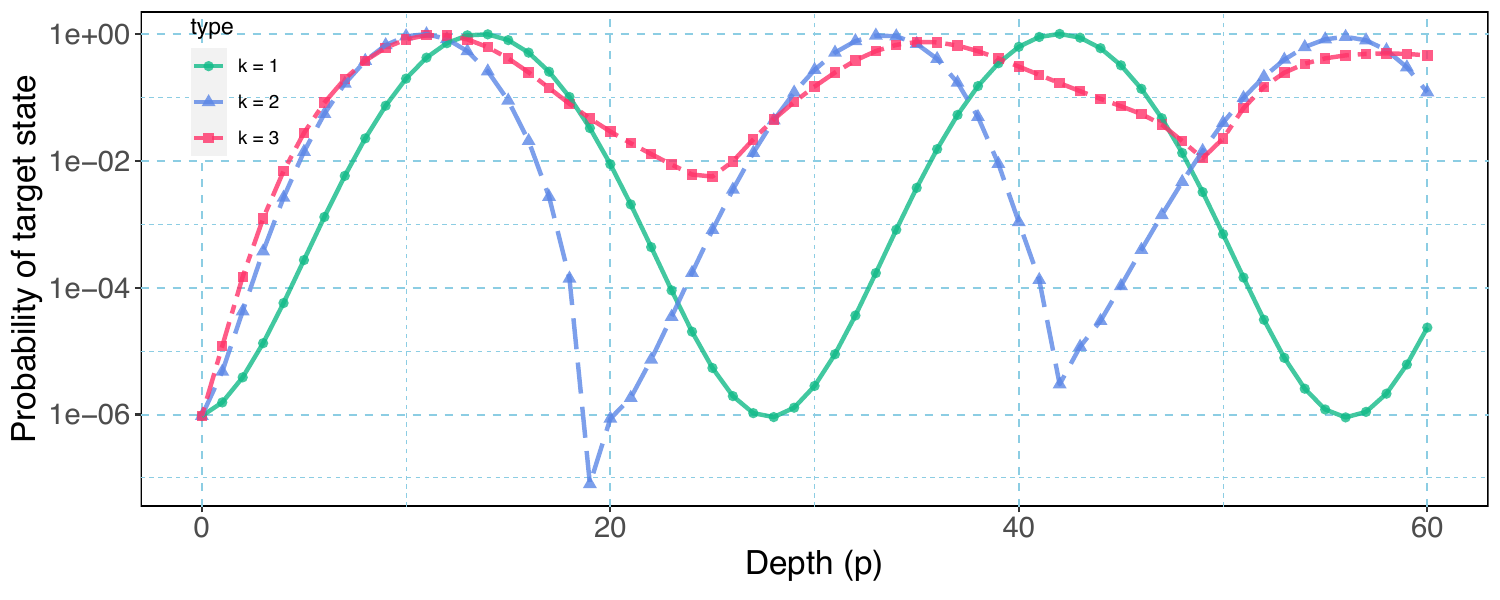}}
		\caption{\small The evolution of the success probability for $k$-local quantum search with $n=20$. The cases of $k=1,2,3$ are represented by the green, blue and red line, respectively.}
		\label{fig_k_local_search}
	\end{figure}
	
	\begin{table}[!t]
		\centering\scriptsize
		\caption{\small The number of iterations $p_k$ required for the $k$-local quantum search to reach the first local maximum of the success probability, along with the corresponding probability.}
		\label{tab_QS}
		\tabcolsep 5pt 
		\begin{tabular*}{1\textwidth}{cccccccccccc}
			\toprule
			$n$ & 10 & 11 & 12& 13 & 14 & 15 & 16 & 17 & 18 & 19 & 20\\
			\hline
			$p_1/\mathrm{Pr}$ & 7/0.958& 7/0.981& 8/1.000& 9/0.983& 9/0.971& 10/0.998& 11/0.995& 12/0.968& 12/0.992& 13/1.000& 14/0.985 \\
			$p_2/\mathrm{Pr}$ & 5/0.987& 6/0.982& 6/0.977& 7/0.993& 7/0.964& 8/0.996& 9/0.965& 9/0.994& 10/0.981& 10/0.988& 11/0.991 \\
			$p_3/\mathrm{Pr}$ & 5/0.921& 6/0.958& 7/0.948& 7/0.955& 8/0.962& 8/0.944& 9/0.966& 10/0.958& 10/0.962& 11/0.967& 11/0.952  \\
			\bottomrule
		\end{tabular*}
	\end{table}
	
	Additionally, the Trotterized circuit for the $k$-local adiabatic quantum search, as shown in Eq.~(\ref{eq_trotterized_k_local_QS}), is applied to 3-local search problem. Denote $T_3(n)$ as the minimal number of iterations required to achieve a success probability exceeding 99\%. For $n$ ranging from 10 to 20, Table~\ref{tab_adiabatic_QS} lists the value $T_3(n)$, along with its corresponding probability. The growth of $T_3(n)$ appears to follow the order of $\mathcal{O}(n^2)$, albeit slightly slower than quadratic growth. If the growth of $T_3$ consistently follows $\Theta(n^2)$, the algorithm’s performance could be further enhanced. 
	
	\begin{table}[!t]
		\centering\scriptsize
		\caption{\small The number of iterations $T_3$ and the corresponding success probability for the Trotterized $k$-local adiabatic search. $T_3$ represents the minimum number of iterations necessary to achieve a success probability exceeding 99\%. }
		\label{tab_adiabatic_QS}
		\tabcolsep 5pt
		\begin{tabular*}{0.817\textwidth}{cccccccccccc}
			\toprule
			$n$ & 10 & 11 & 12& 13 & 14 & 15 & 16 & 17 & 18 & 19 & 20\\
			\hline
			$T_3$ & 98 & 116 & 129 & 143 & 163 & 178 & 201 & 217 & 232 & 259 & 276 \\
			$\mathrm{Pr}$ & 0.9908 & 0.9904 & 0.9906 & 0.9902 & 0.9900 & 0.9904 & 0.9900 & 0.9904 & 0.9900 & 0.9901 & 0.9904 \\
			\bottomrule
		\end{tabular*}
	\end{table}
	
	\newpage
	\appendix
	\titleformat{\section}[block]{\LARGE\bfseries}{Appendices}{1em}{}
	\titlespacing*{\section}{0pt}{*4}{*3} 
	
	\section*{Appendices}
	\addcontentsline{toc}{section}{Appendices}  
	
	\titleformat{\section}[block]{\large\bfseries}{\thesection}{1em}{}
	\renewcommand{\thesection}{\Alph{section}} 
	\titlespacing*{\section}{0pt}{*4}{*2}

	\section{Preliminaries}\label{apsec_prelim}
	
	\subsection{Average-case complexity theory and random \textit{k}-SAT}\label{apsubsec_ac_comp}
	
	In Levin's framework of average-case complexity theory~\cite{Levin1986}, a \textit{random problem} is defined as a pair $(\mu, R)$, where $R\subset \mathbb{N} \times \mathbb{N}$ represents the binary relation on the ``instance-witness'' pair $(x, y)$, and $\mu : \mathbb{N} \to \left[0, 1\right]$ denotes the cumulative distribution function of instances $x$. The probability mass function $\mu'(x)=\mu(x)-\mu(x-1)$ describes the probability of occurrence for a specific instance $x$. A random problem is deemed \textit{polynomial on average} if $\bar{R}(x)\Leftrightarrow \exists yRx$ can be computed in polynomial time with respect to $t(x)$, where $t(x)$ may be exponential in $\left|x\right|$ for some $x$, provided that the ratio $t(x)/\left|x\right|$ remains bounded by a constant on average. Formally, this implies that $\sum_x \mu'(x)t(x)/\left|x\right| < \infty$. The essence of polynomial-on-average complexity lies in its capability to tolerate instances $x$ with extreme difficulty, provided that their occurrence probability $\mu'(x)$ remains sufficiently small.
	
	Given two problems $P_1$ and $P_2$, with instance-witness relations $R_1$ and $R_2$, respectively, a polynomial-time algorithm $f(x)$ reduces $P_1$ to $P_2$ iff $\bar{R}_1(x) \Leftrightarrow \bar{R}_2(f(x))$, where $x$ represents the instance of $R_1$. When dealing with random problems,  both the probability distribution $\mu$ and the instance-witness relation $R$ must be considered. Levin's average-case complexity theory \cite{Levin1986} defines the reduction relationship between $\mu_1$ and $\mu_2$ as $\mu_1 \lesssim \mu_2$ if $\exists k \forall x~\mu'_1(x) / \mu'_2 (x) < \left|x\right|^k$. A polynomial-time algorithm $f(x)$ reduces a random problem $(\mu_1, R_1)$ to $(\mathfrak{f}[\mu_2], R_2)$ if $\mu_1 \lesssim \mu_2$ and $\bar{R}_1(x) \Leftrightarrow \bar{R}_2(f(x))$, where $\mathfrak{f}[\mu_2]$ represents the probability distribution function of output of $f(x)$, defined as $\mathfrak{f}[\mu_2](x') = \sum_{f(x) \le x'} \mu'_2(f(x))$. Based on this reduction framework, a random NP problem is considered complete if every random NP problem is reducible to it. 
	
	In this paper, we focus on the random $k$-SAT problem and its variants. The most commonly used distribution for random $k$-SAT is the model $F(n,m,k)$, where all $k$-SAT instances distribute in uniform probability. Under this model, the random $k$-SAT problem is a random NP-complete problem~\cite{Karp1972, Livne2010}. This owes to the naturalness of $F(n,m,k)$ in describing $k$-SAT instances, which allows reductions from any other random NP problem. By restricting the number of literals in each clause to exactly $k$, $F(n,m,k)$ generates a $k$-SAT instance on $n$ variables by uniformly, independently, and with replacement selecting $m$ clauses from the entire set of $2^kC_n^k$ possible clauses~\cite{Achlioptas2006}. 
	
	The $k$-SAT problem is a decision problem aimed at determining whether a Boolean formula is satisfiable, and it remains NP-complete for $k \ge 3$. Its optimization variant, max-$k$-SAT, seeks the assignment $x$ that satisfies the maximum number of clauses, which is NP-hard for $k \ge 2$. In this paper, we consider a simpler optimization variant---max-$k$-SSAT, which focuses solely on satisfiable instances. This variant belongs to NP, as is can be solved by a non-deterministic Turing machine (NDTM). Although no polynomial-time reduction from max-$k$-SSAT to any existing NP-complete problem has been identified, Lemma~\ref{lem_problem_reduction} demonstrates that its polynomial-time solvability is transferable to the $k$-SAT decision problem. In other words, a polynomial-time solution to max-$k$-SSAT would directly contribute to resolving the P vs NP question. To describe the random max-$k$-SSAT instances, we introduce two random models.
	
	The first model, $F_s(n,m,k)$, is a natural extension of the original $F(n,m,k)$ model, which selectively chooses clauses while maintaining satisfiability. In this process, unsatisfiable branches are directly eliminated. The second model, $F_f(n,m,k)$, is a theoretically motivated approximation of $F_s(n,m,k)$. In $F_f(n,m,k)$, a pre-determined interpretation $t_0$ is randomly provided, and only clauses satisfied by $t_0$ are selected. The details of the clause selecting processes in these random models are depicted in Figure~\ref{fig_clause}. While $F_s(n,m,k)$ offers a natural representation of satisfiable instances, $F_f(n,m,k)$ is more efficient for theoretical analysis. Despite their differences, both models exhibit similar computational properties, as shown in Lemma~\ref{lem_model_reduction}, making $F_f(n,m,k)$ a practical framework for analyzing random satisfiable instances.
	
	\begin{figure}
		\centering
		\small
		{\includegraphics[width=0.48\linewidth]{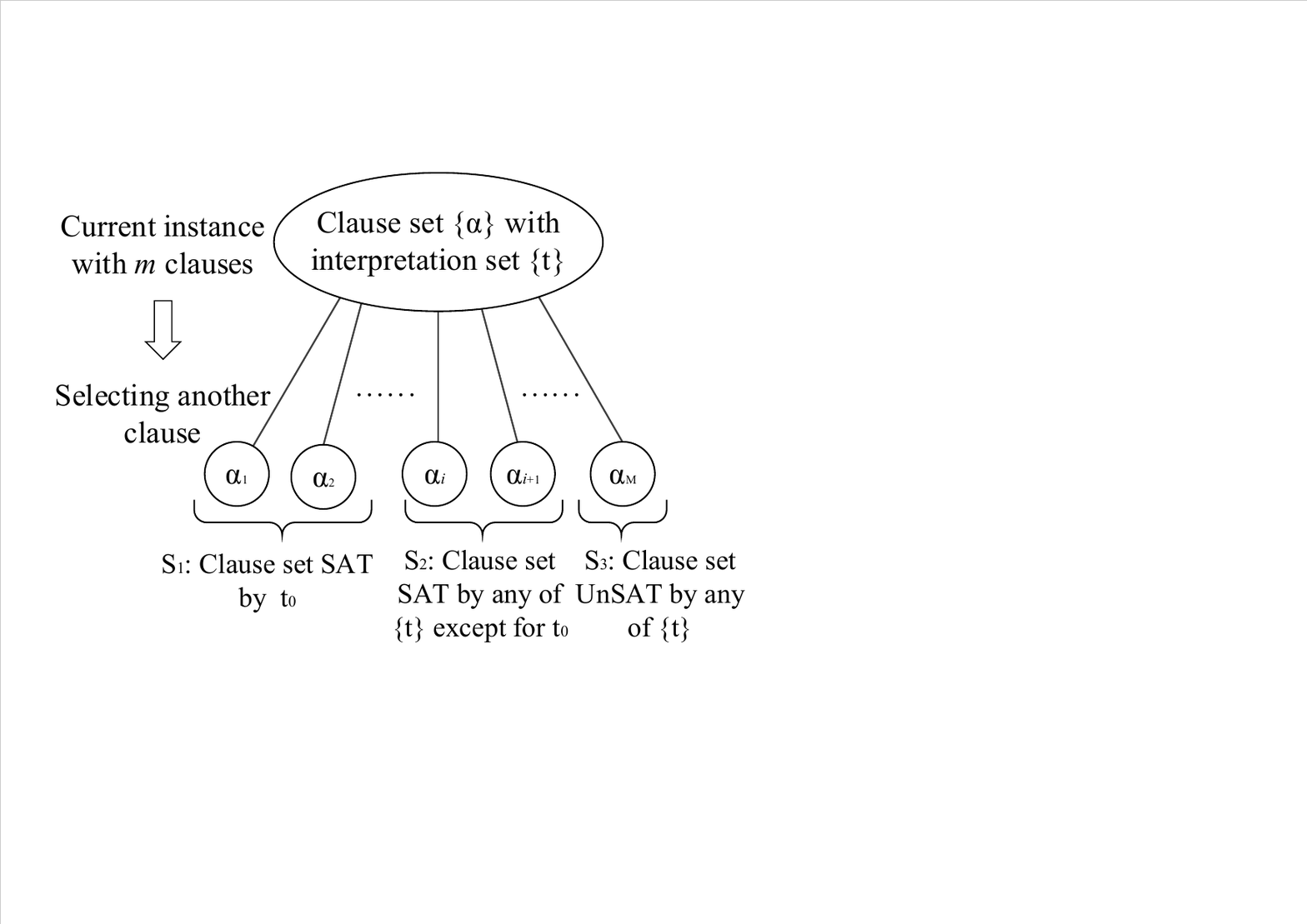}}\small
		\caption{\small The clause selecting process during the generation of a random $k$-SAT instance. Consider an instance $I \in U_s$ with a clause set $\{\alpha\}$ and an interpretation set $\{t\}$. For the model $F(n,m,k)$, arbitrary clause is randomly chosen from $S_1 \cup S_2 \cup S_3$. In contrast, the model $F_s(n,m,k)$ selectively chooses clauses that are satisfiable by any interpretation in $\{t\}$, specifically from $S_1 \cup S_2$. Meanwhile, the model $F_f(n,m,k)$ exclusively selects clauses from the set $S_1$.}
		\label{fig_clause}
	\end{figure}
	
	\subsection{Problem Hamiltonian of \textit{k}-SAT instances}\label{apsubsec_PH_k_SAT}

	In this paper, we introduce a type of multi-controlled phase gate to represent the evolution operator of problem Hamiltonian for solving a Boolean formula. Given a Boolean conjunctive $\alpha = (\lnot) x_{a_1} \land (\lnot)x_{a_2} \land \cdots \land (\lnot)x_{a_k}$, its characteristic function $f_\alpha(x) = 1$ only when the formula is satisfied; otherwise, $f_\alpha(x)=0$. Accordingly, its problem Hamiltonian $h_\alpha$ can be defined as $h_\alpha \left| x \right>= f_\alpha(x) \left| x \right>$. We denote $\alpha$ as $\alpha=(\pm a_1,\pm a_2,\cdots,\pm a_k)$, where $+a_t$ corresponds to $x_{a_t}$ and $-a_t$ to $\lnot x_{a_t}$. The problem Hamiltonian $h_\alpha$ can be expressed as
	\begin{equation}
		{{h}_{\alpha}}=\sigma_{z^\mp}^{(a_1)} \otimes \sigma_{z^\mp}^{(a_2)} \otimes \cdots \otimes \sigma_{z^\mp}^{(a_k)},
	\end{equation}
	where $\pm a_j$ corresponds to $\sigma_{z^\mp}^{(a_j)}$, and the rest positions are tensor-producted with $I$. $\sigma_{z^\pm}^{(a_j)}$ represents $\sigma_{z^\pm}$ on the $a_j$-th qubit, where $\sigma_{z^\pm} = \frac{1}{2}(I\pm\sigma_z)$ are the corresponding components of $\sigma_z$ on $\left| 0 \right>$ and $\left| 1 \right>$, with the matrix form given by
	\begin{equation}
		\sigma_{z^+} = \begin{bmatrix} 1 & 0 \\ 0 & 0 \end{bmatrix}, \quad
		\sigma_{z^-} = \begin{bmatrix} 0 & 0 \\ 0 & 1 \end{bmatrix}.
	\end{equation}
	The evolution operator $e^{i\theta h_\alpha}$ is a $(k-1)$-controlled phase gate with a circuit complexity of ${\mathcal O}(k)$~\cite{Barenco1995}. In single-qubit scenario, it reduces to the phase gate $P_\theta=e^{i\theta\sigma_{z^-}}$, with matrix form given by
	\begin{equation}
		P_\theta = \begin{bmatrix} 1 & 0 \\ 0 & e^{i\theta} \end{bmatrix}.
	\end{equation}
	An example quantum gate for a specific Boolean formula is presented in Figure~\ref{fig_U1gate}.
	
	\begin{figure}
		\centering
		{\includegraphics[width=0.2\linewidth]{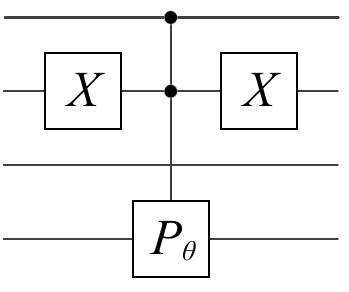}}
		\caption{\small An example quantum circuit for $e^{i\theta h_\alpha}$, where $\alpha = x_1 \land \lnot x_2 \land x_4$. This clause is denoted as $\alpha = (1, -2, 4)$, indicating that the 1st, 2nd, and 4th qubits (from top to bottom in the figure) are involved. An additional pair of $X$ gates is applied to the 2nd qubit due to its inversion. In this circuit, the 4th qubit serves as the target qubit, with the phase gate $P_\theta$ activated on this qubit. The 1st and 2nd qubits are the control qubits, denoted by black point in this figure. However, due to the property of the multi-controlled phase gate, the target qubit can be any of the involved qubits, with the remaining qubits acting as the control qubits.}
		\label{fig_U1gate}
	\end{figure}
	
	In $k$-SAT, a clause $\alpha$ is Boolean disjunctive $\alpha = (\lnot) x_{a_1} \lor (\lnot)x_{a_2} \lor \cdots (\lnot)x_{a_k}$, denoted as $(\pm a_1,\pm a_2,\ldots, \pm a_k)$. By denoting the complement of $\alpha$ as $\bar\alpha = (\mp a_1,\mp a_2,\ldots, \mp a_k)$, it follows that the Hamiltonian of $\alpha$ is $H_\alpha = -h_{\bar\alpha}$. In terms of the whole problem Hamiltonian, since the objective function is $f(x) = \sum_\alpha f_\alpha(x)$, the problem Hamiltonian of the entire Boolean formula is 
	\begin{equation}
		H_C = -\sum_\alpha h_{\bar\alpha}.
	\end{equation}
	As for the evolution operator $e^{-i \theta H_C}$, since $H_C$ is diagonal, $e^{-i \theta H_C}$ can be decomposed as the evolution operator of every single components, expressed as
	\begin{equation}
		e^{-i \theta H_C} = \prod_\alpha e^{-i \theta H_\alpha} .
	\end{equation}
	
	\subsection{Adiabatic quantum computation}

	The adiabatic theorem states that a physical system remains in its instantaneous eigenstate if a given perturbation is acting on it slowly enough and if there is a gap between the eigenvalue and the rest of the Hamiltonian's spectrum~\cite{Born1928}. Building upon this principle, adiabatic quantum computation establishes a gradually evolving system Hamiltonian and encodes the problem's target in the ground state of the final system Hamiltonian, showcasing considerable potential in addressing computationally challenging problems~\cite{farhi2001QAA, Childs2001}. The efficiency of adiabatic quantum computation relies on the required time of the adiabatic evolution, which is determined by the adiabatic approximation conditions~\cite{Tong2007,Wei2007}. 
	
	In the quantum system, the time-dependent Schr{\"o}dinger equation is expressed as 
	\begin{equation}
		i \frac{{\rm d}}{{\rm d} t} \left| \psi(t) \right> = H(t) \left| \psi(t) \right>,
	\end{equation}
	where $H(t)$ represents the time-dependent system Hamiltonian.  Given the initial state $\left| \psi(0) \right>$ in the ground state of $H(0)$ and a gap (difference in eigenvalues) between the ground state of $H(t)$ with the rest eigenstates, $\left| \psi(t) \right>$ will persist in the ground state of $H(t)$ as long as $H(t)$ evolves sufficiently slow.  
	 
	Quantum adiabatic computation~\cite{Albash2018AQC}  leverages the adiabatic theorem to solve for the ground state of a given problem Hamiltonian $H_C$. By designing the system Hamiltonian as 
	\begin{equation}
		H(s) = sH_C +(1-s)H_B, 
	\end{equation}
	and preparing the initial state $\left| \psi(0) \right>$ in the ground state of $H_B$, the state evolves to the ground state of $H_C$ as $s$ slowly varies from 0 to 1. The required time $T$ for the evolution of $H(s)$ should satisfies the condition
	\begin{equation}
		T \gg \frac{\varepsilon_0}{g^{2}_{0}},
	\end{equation} 
	where $g_{0}$ is the minimum of spectral gap between the ground state $\psi_1(s)$ and the first excited state $\psi_2(s)$ of $H(s)$. Denoting $g(s)$ as the spectral gap between $\psi_1(s)$ and $\psi_2(s)$, $g_0=\min_s{g(s)}$. Additionally, $\varepsilon_0$ is determined by the maximum of the derivative of $H(s)$, given by 
	\begin{equation}\label{eq_T_varepsilon}
		\varepsilon_0 = \max_s{\left < \psi_1(s) \left| \frac{\rm d}{{\rm d}s} H(s) \right| \psi_2(s) \right >}. 
	\end{equation}
	
	In simulation, the mapping from  $s\in \left[0, 1\right]$ to time $t\in \left[0,T\right]$ should be established. The Quantum Adiabatic Algorithm (QAA)~\cite{farhi2000QAA, farhi2001QAA} employs a straightforward approach by utilizing a linearly varying system Hamiltonian
	\begin{equation}\label{eq_QAA}
		H(t)=(1-\frac{t}{T})H_B + \frac{t}{T}H_C,
	\end{equation}
	where $0\le t \le T$. Here $H_B$ represents the transverse field $\sum_j \sigma_x^{(j)}$, and the initial state is the superposition state $\left| + \right>^{\otimes n}$. In this case, $\frac{\rm d}{{\rm d}s} H(s)$ is invariant and the evolution time is $\mathcal{O}(g^{-2}_0)$. 
	
	\section{Proof of efficiency on \textit{k}-local search problem}\label{apsec_alg_design}
	
	In this section, we provide the proof for efficiency of the $k$-local quantum search algorithm on $k$-local search problems. We begin by outlining the quantum circuit of the proposed algorithm, including its decomposition and representation. Then, we demonstrate the efficiency for the case $k=1$ with $\theta = \pi$, which leads naturally to conclusion for $\theta = \Theta(n^{-1})$.  For $k \ge 2$, directly proving the case with $\theta = \pi$ is more complex, so we focus on the case $\theta = \Theta(n^{-1})$, which precisely forms the foundation for the further derivation.
	
	\subsection{Circuit for \textit{k}-local quantum search}
	
	The circuit for $k$-local quantum search is presented in Eq.~(\ref{eq_k_local_QS}). To implement the evolution operator $e^{-i\pi H_k}$ on a quantum computer, the Hamiltonian $H_k$ must be decomposed into sub-Hamiltonians that act on a small number of qubits. While the diagonal Hamiltonian $H_k$ can always be decomposed using the Walsh operator \cite{Welch2014}, a more natural decomposition is available in this context. We first consider the decomposition of $H_{k,0}$. Denoting the full set of $k$-combinations from $n$ bits as $I_{n,k}$, for each combination $\alpha = (a_1, a_2,\cdots, a_k) \in I_{n,k}$, the selected bits of $x$ match those of $t=0$ when the Boolean formula $\lnot x_{a_1} \land \lnot x_{a_2} \land \cdots \land \lnot x_{a_k}$ holds. Based on the multi-controlled phase gate discussed in Section \ref{apsubsec_PH_k_SAT}, the Hamiltonian corresponding to this Boolean formula can be denoted as $h_k^{(\alpha)}$, where $h_k^{(\alpha)}$ is $h_k$ applied to the qubits identified by $\alpha=(a_1, a_2,\cdots, a_k)$. Here, $h_k = \sigma_{z^+}^{\otimes k}$. Moreover, $h_k^{(\alpha)}$ represents a component of the overall combination $\alpha$ in the entire $H_{k,0}$, that is, 
	\begin{equation}
		H_{k,0} = \frac{1}{C_n^k} \sum_{\alpha \in I_{n,k}} h_k^{(\alpha)}. 
	\end{equation}
	Since $H_{k,0}$ is diagonal matrix, the evolution operator $e^{-i\theta H_{k,0}}$ can be decomposed into product of each component, expressed as,
	\begin{equation}
		e^{-i\theta H_{k,0}} = \prod_{\alpha \in I_{n,k}} e^{-\frac{i\theta}{C_n^k} h_\alpha}.
	\end{equation} 
	For a general $H_k$, it differs from $H_{k,0}$ only in that these Boolean formulas are satisfied by a general $t$ other than 0. Consequently, for the $j$-th bit, if $t_j$ is not 0, only an additional pair of $X$ gates is required in the $j$-th qubit of the evolution operator $e^{-i\theta h_\alpha}$ to flip the state.
	
	Additionally, the evolution operator of $k$-local quantum search can be expressed in a more concise form as 
	\begin{equation}
		U_k = e^{-i\pi H_{B, k}}e^{-i\pi H_k}.
	\end{equation}
	Here, $H_{B,k}$ generalizes the commonly used mixer Hamiltonian, $H_B = \sum_j \sigma_x^{(j)}$, to $k$-local scenarios, and is given by
	\begin{equation}
		H_{B, k} = \frac{1}{C_n^k} \sum_{\alpha \in I_{n,k}} X_k^{(\alpha)},
	\end{equation}
	where $X_k = H^{\otimes k} h_k H^{\otimes k}$. Lemma \ref{lemma_exchange_H} establishes the equivalence of these two expressions, which reveals that the Hamiltonian $H^{\otimes n}  H_{k,0} H^{\otimes n}$ in Eq.~(\ref{eq_k_local_QS}) is equivalent to $H_{B,k}$. When $k=1$, the evolution operator $e^{i\theta X_k}$ is actually that of $e^{i\theta \sigma_x}$, implying the reduction of $H_{B,k}$ to the transverse field $H_B$ (disregarding the coefficient and global phase). 
	
	\begin{lem}\label{lemma_exchange_H}
		Given a single-qubit gate $M$ that satisfies $MM=I$, the Hamiltonian $\mathcal{H}_{M,k} = M^{\otimes n} (\sum_{\alpha \in I_{n,k}} h_k^{(\alpha)})M^{\otimes n}$ can be expressed as $H_{M,k} =\sum_{\alpha \in I_{n,k}} M_k^{(\alpha)}$, where $M_k^{(\alpha)}$ represents $M_k$ acting on qubits identified by $\alpha$, and $M_k = M^{\otimes k} h_k M^{\otimes k}$.
	\end{lem}
	\begin{proof}
		This lemma is unequivocally established when $k=1$, as the sub-Hamiltonians become local on a single qubit. Besides, this lemma also holds true for the case where $k\geq2$ and $n=k$. Given the establishment of this lemma for the cases of $k-1$ with any arbitrary $n$, as well as for $k$ with a specific $n$, we proceed to demonstrate its validity for the scenario involving $k$ with $n+1$.
		
		To identify the Hamiltonians corresponding to different values of $n$, we introduce an additional subscript $n$ to the notation of the Hamiltonian. The equivalence between $H_{M,n+1,k}$ and $\mathcal{H}_{M,n+1,k}$ can be reduced to establishing the equivalence between evolution operators $U_{n+1,k}(\theta) = e^{i\theta H_{M,n+1,k}}$ and $\mathcal{U}_{n+1,k}(\theta)=e^{i\theta \mathcal{H}_{M,n+1,k}}$ for any arbitrary $\theta$. $U_{n+1,k}(\theta)$ decomposes as 
		\begin{equation*} 
			U_{n+1,k}(\theta)= \prod_{\alpha \in I_{n+1,k}}  e^{i\theta M_k^{(\alpha)}}.
		\end{equation*}
		These terms can be classified based on whether the ($n$+1)-th qubit is evolved. The components excluding the ($n$+1)-th qubit actually constitute $U_{n,k}(\theta)$. For the remaining terms, with the ($n$+1)-th qubit consistently involved, the overall evolution operator can be viewed as $U_{n,k-1}(\theta)$ with an extra control qubit in the ($n$+1)-th position, denoted as $U'_{n,k-1}(\theta)$. As shown in Figure \ref{fig_decom_example}, the decomposition for $U_{n+1,k}$ can be visualized with a specific example where $n=3$ and $k=2$. Consequently, $U_{n+1,k}(\theta)$ can be expressed as
		\begin{equation}\label{eq_n1_n_decompose}
			U_{n+1,k}(\theta) = U_{n,k}(\theta) U'_{n,k-1}(\theta).
		\end{equation}
		
		\begin{figure}
			\centering
			{\includegraphics[width=0.6\linewidth]{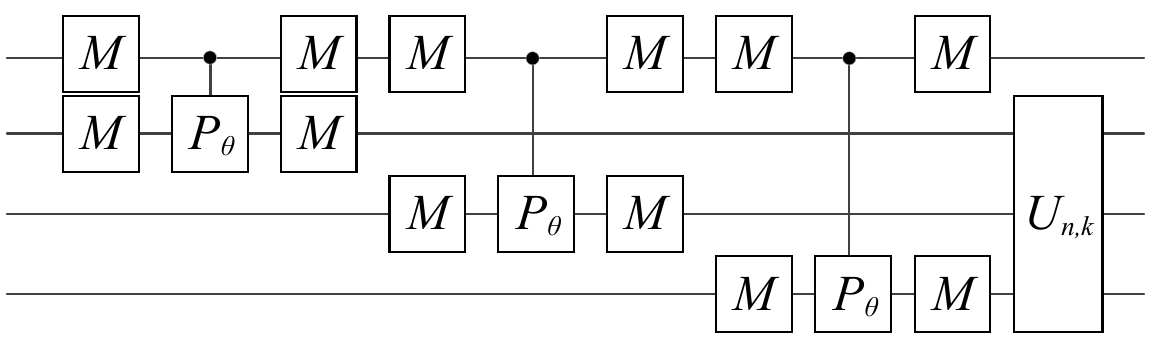}}
			\caption{\small An example for decomposition of evolution operator $U_{n+1,k}(-\theta)$ with $n=3$ and $k=2$.}
			\label{fig_decom_example}
		\end{figure}
		
		Denoting the resulting states of $M$ on computational basis as $\left| +_M \right> = M\left| 0 \right>$ and $\left| -_M \right> = M\left| 1 \right>$, respectively, we illustrate the evolution with respect to the state of the ($n$+1)-th qubit. For an arbitrary ($n$+1)-qubit state $\left| \psi \right>$, it can be rewritten as
		\begin{equation*}
			\left| \psi \right> = \alpha \left| +_M \right>\left| \psi_1 \right> + \beta\left| -_M \right>\left| \psi_2 \right>
		\end{equation*}
		according to the state of the ($n$+1)-th qubit. If the state of the ($n$+1)-th qubit is $\left| +_M \right>$, given $M\left| +_M \right> =\left| 0 \right>$, the control qubit in the ($n$+1)-th position is not satisfiable. Consequently, only $U_{n,k}(\theta)$ acts on the remaining $n$ qubit. Conversely, if the state of the ($n$+1)-th qubit is $\left| -_M \right>$, the evolution on the remaining $n$ qubits should be $U_{n,k}(\theta)U_{n,k-1}(\theta)$. Thus, the evolution can be expressed as
		\begin{equation*}
			U_{n+1,k}(\theta)\left| \psi \right> = \alpha U_{n,k}(\theta)\left| +_M \right>\left| \psi_1 \right> + \beta U_{n,k}(\theta)U_{n,k-1}(\theta)\left| -_M \right>\left| \psi_2 \right>.
		\end{equation*}
		By bringing in the established results, this evolution can be further expressed as 
		$$U_{n+1,k}(\theta)\left| \psi \right> = \alpha\mathcal{U}_{n,k}(\theta)\left| +_M \right>\left| \psi_1 \right> + \beta \mathcal{U}_{n,k}(\theta)\mathcal{U}_{n,k-1}(\theta) \left| -_M \right>\left| \psi_2 \right>,$$
		which exactly represents the evolution operator $\mathcal{U}_{n+1,k}(\theta)$ applied on quantum state $ \left| \psi \right>$. Consequently, the lemma is established for the case of $k$ and $n+1$.
	\end{proof}

	\subsection{Efficiency of 1-local quantum search}\label{apsubsec_QS1_efficiency}
	
	The 1-local scenario serves as the simplest case in $k$-local quantum search, yet it is also the most representative, corresponding to the other extreme case when $k=n$, i.e., Grover's search. In the circuit of 1-local quantum search, the evolution is localized on each single qubit. Consequently, the eigenspace of every 1-local unitary operator comprises the full eigenspace of the search operator, along with its eigenvalues. The iterations in Eq.~(\ref{eq_k_local_QS}) manifest as rotations corresponding to these eigenvalues in the eigenspace. Lemma \ref{lemma_reduce_t} streamlines the circuit of $k$-local quantum search to a unified form by reducing the circuit with arbitrary target $t$ to that with $t=0$. Moreover, Lemma \ref{lemma_eigen_1_local_QS} provides an approximate decomposition of the search operator. Finally, Lemma \ref{lemma_1_local_QS} demonstrates that with $\theta = \pi$ and approximately $p = \frac{n}{\sqrt{2}}$ iterations, the rotation first aligns closely with the target state $\left| t \right>$.
	
	\begin{lem}\label{lemma_reduce_t}
		The circuit of $k$-local quantum search with arbitrary target $t$ can be reduced to that with $t=0$. 
	\end{lem}
	\begin{proof}
		When $k=1$, the Hamiltonian $H_{k,0}$ can be reduced to $H_Z = \sum_j \sigma^{(j)}_{z}$ with normalization, i.e., $H_{k,0}$ is $H_Z/2n$ without considering the global phase. Given that $e^{-i\theta\sigma_z}=X e^{i\theta\sigma_z} X$, $H_k$ with target $t$ is equivalent to $H_{k,0}$ with an additional pair of $X^{t_j}$ gates on both sides of each $j$-th qubit, where $t_j$ is the $j$-th bit of $t$, and $X^1 = X, X^0 = I$. That is, the evolution of 1-local quantum search can be expressed as
		\begin{equation*}
			\left| \psi_p \right> =  {\left( H^{\otimes n} e^{-i\pi H_Z/2n} H^{\otimes n} X_C e^{-i\pi H_Z/2n} X_C\right) }^p  \left| + \right>^{\otimes n},
		\end{equation*}
		where $X_C = \prod_j X_{(j)}^{t_j}$. $X_{(j)}^{t_j}$ is $X^{t_j}$ applied on the $j$-th qubit, and $X^{t_j}$ is the $t_j$-th power of gate $X$. It is simple to verify that  $H{{e}^{i\theta \sigma_z}}H$ commutes with $X^{t_j}$. Therefore, the most of $X_C$ are canceled, and the evolution is reduced to 
		\begin{equation}\label{eq_QS1_reduced_circuit}
			\left| {{\psi}_{p}} \right\rangle = X_C {\left( H^{\otimes n} e^{i\pi H_Z/2n} H^{\otimes n} e^{i\pi H_Z/2n}  \right) }^p  \left| + \right>^{\otimes n},
		\end{equation}
		with only a single ${{X}_{C}}$ remaining. The operator ${{X}_{C}}$ fully characterizes the information of $\left| t \right\rangle $ due to the relation ${{X}_{C}}\left| 0 \right\rangle =\left| t \right>$. Consequently, the circuit of 1-local search with arbitrary target $t$ is reduce to that with $t=0$.	
		
		When $k\ge2$, this lemma is also established. Initially, for $n=k$, the lemma holds true. Given the property $X_{(j)}^{t_j}X_{(j)}^{t_j} = I$, the additional pair of gate $X_{(j)}^{t_j}$ on every component $e^{i\theta h_k^{(\alpha)}}$ can be extracted according to Lemma \ref{lemma_exchange_H}. Therefore, the evolution of $k$-local quantum search with target $t$ can be reformulated as 
		\begin{equation*}
			\left| \psi_p \right> =  {\left( H^{\otimes n} e^{-i\pi H_k} H^{\otimes n} X_C e^{-i\pi H_k} X_C\right) }^p  \left| + \right>^{\otimes n}.
		\end{equation*}
		Moreover, Lemma~\ref{lemma_HX_commutator} demonstrates that $H^k e^{i\theta h_k} H^k$ commutes with the $X$ gate. By following a similar line of reasoning, the same conclusion can be reached for the circuit of $k$-local quantum search.
	\end{proof}
	
	\begin{lem}\label{lemma_HX_commutator}
		For any arbitrary $k$, the operator $H^k e^{i\theta h_k} H^k$ commutes with the $X$ gate, irrespective of upon which qubit the $X$ gate operates.
	\end{lem}
	\begin{proof}
		When $k=1$, the operator reduces to $H e^{i\theta \sigma_z/2}H$ without considering the global phase. Evaluating the establishment of this lemma in this case is straightforward. Having established that this lemma holds true for a specific $k$, we proceed to prove its correctness for the case of $k+1$. In fact, ${{e}^{i\theta h_{k+1}}}$ can be regarded as a controlled $e^{i\theta h_k}$ with an extra control qubit at the ($k$+1)-th qubit. Thus, when the $X$ gate is applied to any of the first $k$ qubits, the commutativity is assured. Here, our focus lies in the scenario when the $X$ gate is on the ($k$+1)-th qubit.
		
		To establish the commutativity of $H^k e^{i\theta h_k} H^k$ with gate $X^{(k+1)}$, it suffices to demonstrate the equivalence between $H^k e^{i\theta h_k} H^k$ and  $X^{(k+1)} H^k e^{i\theta h_k} H^k X^{(k+1)}$, where $X^{(k+1)}$ is $X$ gate applied on the ($k$+1)-th qubit. For any arbitrary ($k$+1)-qubit state $\left| \psi \right>$,  it can be expressed as 
		\begin{equation*}
			\left| \psi \right> = \alpha \left| + \right>\left| \psi_1 \right> + \beta\left| - \right>\left| \psi_2 \right>
		\end{equation*}
		with respect to the state of the ($k$+1)-th qubit. Upon applying the gate $X^{(k+1)}$, the additional $X$ gate on the ($k$+1)-th qubit flips the phase when the state of the ($k$+1)-th qubit is $\left| - \right>$, that is,
		\begin{equation}
			X^{(k+1)} \left| \psi \right> = \alpha \left| + \right>\left| \psi_1 \right> - \beta\left| - \right>\left| \psi_2 \right>. 
		\end{equation}
		The resulting state after the evolution of $X^{(k+1)} H^{\otimes k+1} e^{i\theta h_{k+1}} H^{\otimes k+1} X^{(k+1)}$ is
		\begin{equation*}
			X^{(k+1)} H^{\otimes k+1} e^{i\theta h_{k+1}} H^{\otimes k+1} X^{(k+1)} \left| \psi \right> = \alpha \left| + \right> \left| \psi_1 \right> + \beta\left| - \right> H^{\otimes k}  e^{i\theta h_k} H^{\otimes k}\left| \psi_2 \right>, 
		\end{equation*}
		which coincides with the resulting state after applying $H^{\otimes k+1} e^{i\theta h_{k+1}} H^{\otimes k+1}$ to $\left| \psi \right>$. Consequently, this lemma is established.
	\end{proof}
	
	
	\begin{lem}\label{lemma_eigen_1_local_QS}
		The 1-local search operator $U_C = H^{\otimes n} e^{i\pi H_Z/2n} H^{\otimes n} e^{i\pi H_Z/2n}$ can be approximately eigendecomposed as
		\begin{equation}
			U_C = {{V}^{T}}EV + \mathcal{O}(n^{-2}), 
		\end{equation}
		where $V=V_{0}^{\otimes n}$, and $E=E_{0}^{\otimes n}$ with 
		\begin{equation}
			{{E}_{0}}=\left[ \begin{matrix}
				{{e}^{\frac{i\pi }{\sqrt{2}n}}} & 0  \\
				0 & {{e}^{\frac{-i\pi }{\sqrt{2}n}}}  \\
			\end{matrix} \right],
			{{V}_{0}}=\left[ \begin{matrix}
				\cos \left( \frac{\pi }{8} \right) & \sin \left( \frac{\pi }{8} \right)  \\
				-\sin \left( \frac{\pi }{8} \right) & \cos \left( \frac{\pi }{8} \right)  \\
			\end{matrix} \right].
		\end{equation}
	\end{lem}
	\begin{proof}
		In the context of 1-local quantum search, the entire evolution operator can be decomposed into components on individual qubits, described by $U_0 = H{{e}^{\frac{i\pi }{2n}Z}}H{{e}^{\frac{i\pi }{2n}Z}}$ for $p$ iterations. The operator $U_0$ can be represented as 
		\begin{equation*}
			U_0=\frac{1}{2} \left[ \begin{matrix}
				e^{i\pi/n}+1& 1-e^{i\pi/n} \\
				e^{i\pi/n}-1 & 1+e^{i\pi/n} \\
			\end{matrix} \right]. 
		\end{equation*}
		Applying the Taylor series expansion, $U_0$ can be further approximated as
		\begin{equation*}
			U_0=\left[ \begin{matrix}
				1+i\frac{\pi }{2n}-{{\left( \frac{\pi }{2n} \right)}^{2}}  & i\frac{\pi }{2n}+{{\left( \frac{\pi }{2n} \right)}^{2}}  \\
				i\frac{\pi }{2n}-{{\left( \frac{\pi }{2n} \right)}^{2}} & 1-i\frac{\pi }{2n}-{{\left( \frac{\pi }{2n} \right)}^{2}}  \\
			\end{matrix} \right] + \mathcal{O}(n^{-3}),
		\end{equation*}
		We denote $a=\cos \left( {\pi }/{8} \right)$ and $b=\sin \left( {\pi }/{8} \right)$. Regarding the main component of $U_0$, $u = V_{0}^{T}{{E}_{0}}{{V}_{0}}$, which can be expanded as 
		\begin{equation*}
			u = \left[ \begin{matrix}
				a^2 e^{i\pi/\sqrt2n} + b^2 e^{-i\pi/\sqrt2n} & ab (e^{i\pi/\sqrt2n}  - e^{-i\pi/\sqrt2n})  \\
				ab (e^{i\pi/\sqrt2n}  - e^{-i\pi/\sqrt2n}) & b^2 e^{i\pi/\sqrt2n} + a^2 e^{-i\pi/\sqrt2n}  \\
			\end{matrix} \right],
		\end{equation*}
		The remainder term $\Delta u = U_0 - u$ should have a magnitude of
		\begin{equation*}
			\Delta u=\left[ \begin{matrix}
				\mathcal{O} ( {{n}^{-3}} ) & \mathcal{O} ( {{n}^{-2}} )  \\
				\mathcal{O} ( {{n}^{-2}} ) & \mathcal{O} ( {{n}^{-3}} )  \\
			\end{matrix} \right].
		\end{equation*}
		
		Denoting $\Delta U = U_C-u^{\otimes n}$, directly stating $\left\| \Delta u \right\| = \mathcal{O}(n^{-2})$ would imply $\left\| \Delta U \right\| = \mathcal{O}(n^{-1})$, which fall short of meeting the requirements of the subsequent proof. Here, a more nuanced analysis delves into the search operator $U_C$, revealing that $\left\| \Delta U \right\| = \mathcal{O}(n^{-2})$. Upon examination of both matrices $u$ and $\Delta u$, it is apparent that the diagonal and non-diagonal elements exhibit distinct orders of magnitude. Specifically, for the diagonal elements of $u$ and $\Delta u$, we denote ${u_1}= \mathcal{O}(1)$ and $\Delta {u_1}= \mathcal{O}({{n}^{-3}})$. Correspondingly, for the non-diagonal elements, we denote ${{u}_{0}} = \mathcal{O}(n^{-1})$ and $\Delta {{u}_{0}}= \mathcal{O}({{n}^{-2}})$.
		
		
		
		During the tensor-production process of $U_C = U_0^{\otimes n}$, the elements of the matrix ${{U}_{C}}$ can be identified by an $n$-bit binary string $x$, where $x_j = 0$ or $1$ represents whether the non-diagonal or diagonal element is selected in the $j$-th position, respectively. Specifically, the element of $U_C$ with identifier $x$ is expressed as
		\begin{equation*}
			U_C(x)=\prod\limits_{j=1}^{n}{\left( {{u}_{{{x}_{j}}}}+\Delta {{u}_{{{x}_{j}}}} \right)}.
		\end{equation*}
		With another identifier $y$, $U_C(x)$ can be further expanded into $N$ terms, expressed as
		\begin{equation}\label{eq_QS_P1}
			{{U}_{C}}(x)=\sum\limits_{y=0}^{N-1}{\left( \prod\limits_{j=1}^{n}{u_{x_j}^{1-y_j}\Delta u_{x_j}^{y_j}} \right)},
		\end{equation}
		where $y_j=0$ or $1$ indicates whether the major or minor component is multiplied in the $j$-th position, respectively.
		Denoting the $N$ terms in the summation of Eq.~(\ref{eq_QS_P1}) as $f(x,y)$, for a fixed $x$, the magnitude of $f(x,y)$ can be classified by $l={{d}_{H}}(y,0)$, where $l$ represents the possible number of $\Delta u_{x_j}$ selected. The terms of $f(x,y)$ with $l=0$ constitute the main component $u^{\otimes n}$, and the remaining residue is given by
		\begin{equation*}
			\Delta U(x) = \sum_{l>0} \Delta U(x, l),
		\end{equation*} 
		where 
		\begin{equation*}
			\Delta U(x, l) = \sum_{{{d}_{H}}(y,0)=l} f(x,y).
		\end{equation*}
		$\Delta U(x)$ can also be classified based on $m={{d}_{H}}(x,0)$, where $m$ indicates the number of selected diagonal elements. For $\Delta U(x)$ with the same $m$, the magnitude remains consistent, denoted as	\begin{equation*}
			\Delta U(m) = \sum_{l>0} \Delta U(m, l),
		\end{equation*} 
		where $\Delta U(m, l)$ represents a specific $\Delta U(x, l)$ with ${{d}_{H}}(x,0)=m$.	
		
		In the following proof, we demonstrate that for every value of $m$, $\Delta U(m) = \mathcal{O} (n^{-2})$. Initially, for the special case when $m=n$, every selected element in $f(x,y)$ is on the diagonal position. Consequently, the magnitude of $\Delta U(n,l)$ is $\mathcal{O}(C_n^l n^{-3l})$, resulting in an overall summation of $\Delta U(n)$ as $\mathcal{O}(n^{-2})$. We then explore the more general scenario when $0\le m<n$ and $0<l\le n$. Specifically, in cases of $l+m \le n$ and $l+m>n$, the magnitude of $\Delta U(n)$ is still $\mathcal{O}(n^{-2})$. 
		
		When $l+m \le n$, the maximal elements of $f(x,y)$ within $\Delta U(m,l)$ consist of $m$ diagonal elements $u_{1}$ and $n-l-m$ non-diagonal elements $u_{0}$. The remaining $l$ positions are filled with non-diagonal elements $\Delta u_{0}$. Consequently, the magnitude of these maximum terms is $\mathcal{O}({n^{-(n+l-m)}})$. Given that there are $C_{n-m}^{l}$ such terms, the total summation of these maximum terms is of the order $\mathcal{O}(n^{-2})$. For non-maximum terms, where fewer diagonal elements $u_1$ are selected, the magnitude of ${f}(x,y)$ is even smaller, contributing significantly less to the total.
		
		The conclusion holds when $l+m>n$ as well. In this case, the maximal elements of $f(x,y)$ within $\Delta U(m,l)$ consist of  $n-l$ diagonal elements $u_1$. The remaining elements are constituted by $\Delta u_{x_j}$, with $m-n+l$ diagonal elements $\Delta u_1$ and $n-m$ non-diagonal elements $\Delta u_0$. As a result, the magnitude of these maximums is $\mathcal{O}({n^{-(m+3l-n)}})$. With $C_m^{n-l}$ such terms, their summation remains of the order $\mathcal{O}(n^{-2})$. The contribution from the non-maximum terms is again negligible for similar reasons. Thus, each element of $\Delta {U}$ is asymptotically bounded by $\mathcal{O}({n^{-2}})$. 
	\end{proof}
	
	\begin{lem}\label{lemma_1_local_QS}
		For $k=1$, the $k$-local quantum search requires approximately $\frac{n}{\sqrt{2}}$ Oracle calls to evolve the initial state towards the target state $\left| t \right>$, where the amplitude of $\left| t \right>$ converges to 1 with the increase of $n$.
	\end{lem}
	\begin{proof}
		According to Lemma \ref{lemma_eigen_1_local_QS}, the search operator can be approximately decomposed as ${{V}^{T}}EV$, and the whole iteration of 1-local quantum search is ${V^T}{E^p}V$. According to the evolution of 1-local quantum search presented in Eq.~(\ref{eq_QS1_reduced_circuit}), the proof is reduced to determine whether there exists a certain $p$ such that $V\left| 0 \right\rangle \approx {{E}^{p}}VH^{\otimes n} \left| 0 \right\rangle$. Denoting $a=\cos \left( {\pi }/{8} \right)$, $b=\sin \left( {\pi }/{8} \right)$, it is straightforward to verify that $\sqrt{2}a = a+b$ and $\sqrt{2}b = a-b$. Bringing in $V$, the state $V\left| 0 \right\rangle$ and $VH\left| 0 \right\rangle$ can be expanded as 
		\begin{equation*}
			V\left| 0 \right\rangle =\sum\limits_{x=0}^{N-1}{{{\left( -1 \right)}^d}{{a}^{n-d}}{b^d}\left| x \right\rangle }, \; VH\left| 0 \right\rangle =\sum\limits_{x=0}^{N-1}{a^{n-d}b^d\left| x \right\rangle }.
		\end{equation*}
		where $d=d_H \left(x, 0 \right)$. Therefore, the phase difference between $V\left| 0 \right\rangle $ and $VH\left| 0 \right\rangle $ for the computation basis $\left| x \right>$ is ${{\left( -1 \right)}^d}={{e}^{id\pi }}$. The phase shift of $\left| x \right\rangle $ after the operator $E$ depends on $d$, and after $p$ iterations, the overall phase shift is approximately ${{e}^{i\left( n-2d \right)p\pi /\sqrt{2}n}}$. Regardless of the global phase, the condition is first satisfied when $p = \frac{n}{\sqrt{2}}$. Thus, the required number of iterations is proved.
		
		For the convergence, we focus primarily on the deviation $\Delta {{U}_{p}} = U_{C}^{p}-{{V}^{T}}{{E}^{p}}V$. The main contribution to $\Delta {{U}_{p}}$ comes from the first-order terms involving $\Delta U$. Consequently, the iterations of the search operator can be expanded as
		\begin{equation*}
			U_{C}^{p}= {{V}^{T}}E^pV + \sum_{j=1}^{p} {V^T}E^{j-1} V \Delta U  {V^T}E^ {p-j}V + \mathcal{O}(n^{-1}). 
		\end{equation*}
		The resulting state after applying $U_{C}^{p}$ to $H^{\otimes n} \left| 0 \right>$ can be expressed as
		\begin{equation*}
			\left| {{\psi}_{p}} \right\rangle ={{V}^{T}}{{E}^{p}}VH\left| 0 \right\rangle +\Delta {{U}_{p}}H\left| 0 \right>.
		\end{equation*}
		With $p=\mathcal{O}(n)$, each element of $\Delta U_p$ has a magnitude of $\mathcal{O}(n^{-1})$. As a result, its influence on the final result, $\left<0\right|\Delta {{U}_{p}}H\left| 0 \right\rangle$, remains within the magnitude of $\mathcal{O}(n^{-1})$. Thus, the convergence is established.
	\end{proof}
	
	\subsection{Efficiency of \textit{k}-local quantum search with a small \textit{k}}\label{subsec_QSk_efficiency}
	
	For a general $k$, the objective function of the $k$-local quantum search gradually varies as $k$ increases. Specifically, considering the deviation $\Delta f_{k}(x) = f_{k+1}(x) - f_{k}(x)$, it can be expressed as
	\begin{equation}\label{eq_deviation_fkx}
		\Delta f_{k}(x) = -\frac{n-d}{n-k}f_k(x),
	\end{equation}
	where $d  = n- d_H(x,t)$, and $\Delta f_{k}(x) \le 0$ for any $k<n$ and arbitrary $x$. Obviously, the deviation $\Delta f_{k}(x)$ represents the gradual loss of structural information about the search problem. When $k=1$, the complete structural information enables the efficient resolution of the search problem, as established in Section \ref{apsubsec_QS1_efficiency}. As $k$ slightly increase, the search algorithm remains effective; however, the reduction in structural information increases the difficulty of solving the problem, leading to an increase in the required number of iterations (or a reduction in the success probability). Eventually, when $k \to n$, the $k$-local quantum search reduces to the Grover's search, requiring approximately $\frac{\pi}{4}\sqrt{N}$ iterations. 
	
	The cases of $k=1$ and $k=n$ establish the lower and upper bounds in magnitude for the required number of iterations. In this paper, our analysis primarily focuses on scenarios with a small constant $k$. Given the simplicity of $k$-local search problem on classical computers and the proven efficiency of quantum algorithms \cite{Benioff1980, Barenco1995}, it is reasonable to infer that the $k$-local search problem is polynomially solvable on quantum computers, which will be demonstrated through the analysis below.
	
	Specifically, for a small constant $k \ge 2$, the whole evolution can also be conceptualized as rotation within the eigenspace. Thus, it is feasible to consider the $k$-local quantum search with a small rotation angel, expressed as
	\begin{equation}
		\left| \psi \right> = \left( e^{-i\theta H_{B,n,k}}e^{-i\theta H_{n,k}} \right)^{p_\theta} \left|+ \right>,
	\end{equation}
	where the magnitude of coefficient $\theta$ is constrained to $\mathcal{O}(n^{-1})$. Here, demonstrating the essential number of iterations $p$ to be $\mathcal{O}(n)$ is reduced to establishing that the required $p_\theta$ is $\mathcal{O}(n^2)$. Moreover, according to Trotter-Suzuki decomposition that 
	\begin{equation}
		e^{-i\theta \mathcal{H}_{n,k}}  = e^{-i\theta/2 H_{n,k}}e^{-i\theta H_{B,n,k}}e^{-i\theta/2 H_{n,k}} + \mathcal{O}(\theta^{3}),
	\end{equation}
	the evolution of $k$-local quantum search can be approximated as 
	\begin{equation}\label{eq_approx_QS1}
		\left( e^{-i\theta H_{B,n,k}}e^{-i\theta H_{n,k}} \right)^{p_\theta} \left|+ \right> =  e^{-i\theta p_\theta \mathcal{H}_{n,k}}  \left| + \right> + \mathcal{O}(n^{-1}) +\mathcal{O}(p_\theta n^{-3}), 
	\end{equation}
	where the Hamiltonian $\mathcal{H}_{n,k} = H_{B,n,k} + H_{n,k}$. Under this condition, Our proof revolves around establishing the spectral gap $g_k$ of $\mathcal{H}_{n,k}$ to be $\Theta(n^{-1})$. In this context, with the target encoded in the computational basis state of $H_C$ with the highest energy, the spectral gap refers to the energy gap between the eigenstate with the highest and second-highest energy in $\mathcal{H}_{n,k}$.
	
	In the subsequent proof, Lemmas \ref{lemma_commutator} and \ref{lemma_eigen_reduction} lay the groundwork for the analysis. Based on these, Lemma \ref{lemma_gap} rigorously demonstrates that the magnitude of the spectral gap in the Hamiltonian $\mathcal{H}_{n,k}$ is $\Theta(n^{-1})$. Accordingly, Lemma \ref{lemma_k_local_QS} shows that the number of iterations required for $k$-local quantum search with $\theta = \Theta(n^{-1})$ remains in the order of $\mathcal{O}(n^2)$. 
	
	\begin{lem}\label{lemma_commutator}
		For any small $\theta_1,\theta_2$, when the qubits affected by $e^{i\theta_1 P_{k_1}}$ overlaps with those affected by $e^{i\theta_2 X_{k_2}}$, the operator $e^{i\theta_1 P_{k_1}}$ can approximately commutate with $e^{i\theta_2 X_{k_2}}$, with a deviation of ${\mathcal O}(\theta_1 \theta_2)$.
	\end{lem}
	\begin{proof}
		For fixed values of $k_1$ and $k_2$, as the number of overlapping qubits increases, the deviation is expected to grow. Consequently, for the upper bound scenario, we assume $k_1 = k_2$ and that all the affected qubits are identical. Given that 
		$$e^{i\theta_2 X_k}e^{i\theta_1 P_k} = e^{i\theta_1 P_k}e^{-i\theta_1 P_k}e^{i\theta_2 X_k}e^{i\theta_1 P_k},$$ 
		the commutator 
		$$\left[e^{i\theta_1 P_{k}}, e^{i\theta_2 X_{k}}\right] = e^{i\theta_1 P_{k}}e^{i\theta_2 X_{k}} - e^{i\theta_2 X_{k}}e^{i\theta_1 P_{k}}$$
		is in the same magnitude with 
		$$\Delta E = e^{-i\theta_1 P_k} e^{i\theta_2 X_k} e^{i\theta_1 P_k}-e^{i\theta_2 X_k}.$$ 
		For an arbitrary square matrix $M$, the deviation between $e^{-i\theta_1 P_k}Me^{i\theta_1 P_k}$ and $M$ lies solely in the non-diagonal elements $\{M_{K,m}(e^{-i\theta_1}-1) | 1 \le m < K \}$ and $\{M_{l,K}(e^{i\theta_1}-1) | 1 \le l < K \}$, where $K=2^k$. The subscript of $M_{l,m}$ represents the position of the deviation in the matrix. Regarding $e^{i\theta_2 X_k}$, since the deviation 
		$$e^{i\theta_2 X_k} - I= H^{\otimes k} (e^{i\theta_2 P_k}-I) H^{\otimes k},$$ 
		is of order ${\mathcal O}(\theta_2)$ for every element, it follows that every non-diagonal element of $e^{i\theta_2 X_k}$ is also of order ${\mathcal O}(\theta_2)$. Combining these approximations, we find that every element of $\Delta E$ is ${\mathcal O}(\theta_1\theta_2)$, and therefore, the same magnitude applies to the commutator $\left[e^{i\theta_1 P_{k}}, e^{i\theta_2 X_{k}}\right]$.
	\end{proof}

	\begin{lem}\label{lemma_eigen_reduction}
		For a sufficiently large $n$, the eigenvalues of $\mathcal{H}_{n+1,k}$ are of the same order of magnitude as those of $\tilde{\mathcal{H}}_{n+1,k}$, where
		\begin{equation}\label{eq_approxi_Hn}
			\tilde{\mathcal{H}}_{n+1,k} = \left[ \begin{matrix}
				\frac{n+1-k}{n+1}\mathcal{H}_{n,k} &  \\
				& 	\frac{n+1-k}{n+1}\mathcal{H}_{n,k} + \frac{k}{n+1}\mathcal{H}_{n,k-1}  \\
			\end{matrix} \right].
		\end{equation}
	\end{lem}
	\begin{proof}
		
		With $\theta$ on the order of $\Theta(n^{-1})$, the analysis of the gap of $\mathcal{H}_{n,k}$ can be simplified by examining the gap of Hamiltonian of $U_{n+1,k}(\theta) = e^{-i\theta H_{B,n+1,k}} e^{-i\theta H_{n+1,k}}$. Moreover, as illustrated in Eq.~(\ref{eq_n1_n_decompose}), both $H_{n+1,k}$ and $H_{B,n+1,k}$ can be decomposed based on whether the ($n$+1)-qubit is evolved, with the only difference lying in the additional coefficient of $H_{n+1,k}$ and $H_{B,n+1,k}$, ie, $1/C_{n+1}^k$. Taking $H_{n+1,k}$ as an example, it can be expressed as
		\begin{equation*}
			H_{n+1,k} = \frac{n+1-k}{n+1}I\otimes H_{n,k} + \frac{k}{n+1}H'_{n,k-1},
		\end{equation*}
		where $H'_{n,k-1}$ represents $H_{n,k-1}$ with an extra control qubit at the ($n$+1)-th position during the evolution. Denoting $H'_{B,n,k-1} = H^{\otimes n+1} H'_{n,k-1} H^{\otimes n+1}$, $U_{n+1,k}(\theta)$ can be expanded as 
		\begin{equation*}
			U_{n+1,k}(\theta) = e^{-i\theta\frac{n+1-k}{n+1} H_{B,n,k}} e^{-i\theta\frac{k}{n+1} H'_{B,n,k-1}}  e^{-i\theta\frac{k}{n+1} H'_{n,k-1}} e^{-i\theta\frac{n+1-k}{n+1} H_{n,k}}.   
		\end{equation*}
		There are $C_n^{k-1}$ terms in $e^{-i\theta\frac{k}{n+1} H'_{B,n,k-1}} $ and $e^{-i\theta\frac{k}{n+1} H'_{n,k-1}}$, each associated with a coefficient $1/C_{n}^{k-1}$. Due to the factor $\frac{k}{n+1}$, there two operators can commute with each other with a deviation on the order of ${\mathcal O}(\theta^2 n^{-2})$, as established in Lemma \ref{lemma_commutator}. Consequently, the evolution is reduced to 
		\begin{equation*}
			U_{n+1,k}(\theta) = e^{-i\theta\frac{n+1-k}{n+1} H_{B,n,k}} e^{-i\theta\frac{k}{n+1} H'_{n,k-1}} e^{-i\theta\frac{k}{n+1} H'_{B,n,k-1}} e^{-i\theta\frac{n+1-k}{n+1} H_{n,k}} + {\mathcal O}(n^{-4}).
		\end{equation*}
		
		Regardless of the negligible deviation in the operator, we denote the dominant component of $U_{n+1,k}(\theta)$ as $U_0(\theta)$. $U_0(\theta)$ is constituted by two analogous components, that is, $U_0(\theta) = U_{H_1, n+1, k}(\theta)U_{H_2, n+1, k}(\theta)$, respectively expressed as 
		\begin{align*}
			& U_{H_1, n+1, k}(\theta) = H^{\otimes n+1} e^{-i\theta\frac{n+1-k}{n+1} H_{n,k} }H^{\otimes n+1} e^{- i\theta\frac{k}{n+1} H'_{n,k-1}}, \nonumber\\
			&  U_{H_2, n+1, k}(\theta) = H^{\otimes n+1} e^{-i\theta\frac{k}{n+1} H'_{n,k-1}} H^{\otimes n+1} e^{-i\theta\frac{n+1-k}{n+1} H_{n,k}}.
		\end{align*}
		For an arbitrary computational basis $\left| x \right>$, the evolution of $U_{H_1, n+1, k}(\theta)$ on the states $\left| 0 \right> \left| x \right>$ and $\left| 1 \right> \left| x \right>$ are expressed as 
		\begin{align*}
			& U_{H_1, n+1, k}(\theta)  \left| 0 \right> \left| x \right> = e^{-i\theta\frac{n+1-k}{n+1} H_{B, n,k} }  \left| 0 \right> \left| x \right>, \nonumber\\
			& U_{H_1, n+1, k}(\theta)  \left| 1 \right> \left| x \right> = e^{-i\theta\frac{n+1-k}{n+1} H_{B, n,k} } e^{- i\theta\frac{k}{n+1} H_{n,k-1}} \left| 1 \right> \left| x \right>. 
		\end{align*}
		Thus, the evolution operator $U_{H_1, n+1, k}(\theta)$ can be reformulated in block-matrix form as
		\begin{equation*}
			U_{H_1, n+1, k}(\theta) = \left[ \begin{matrix}
				e^{-i\theta\frac{n+1-k}{n+1} H_{B, n,k} } &  \\
				& e^{-i\theta\frac{n+1-k}{n+1} H_{B, n,k} } e^{-i\theta\frac{k}{n+1} H_{n,k-1}}  \\
			\end{matrix} \right]. 
		\end{equation*}
		Similarly, the evolution of $H^{(n+1)} U_{H_2, n+1, k}(\theta)H^{(n+1)}$, where $H^{(n+1)}$ represents the Hadamard gate on the ($n$+1)-th qubit, can be expressed as 
		\begin{align*}
			& H^{(n+1)} U_{H_2, n+1, k}(\theta) H^{(n+1)} \left| 0 \right> \left| x \right> = e^{-i\theta\frac{n+1-k}{n+1} H_{n,k} }  \left| 0 \right> \left| x \right>,\\
			& H^{(n+1)} U_{H_2, n+1, k}(\theta) H^{(n+1)} \left| 1 \right> \left| x \right> = e^{- i\theta\frac{k}{n+1} H_{B, n,k-1}} e^{-i\theta\frac{n+1-k}{n+1} H_{n,k} }  \left| 1 \right> \left| x \right>, \nonumber
		\end{align*}
		which can also be represented in block-matrix form as
		\begin{equation*}
			H^{(n+1)}U_{H_2, n+1, k}(\theta)H^{(n+1)}= \left[ \begin{matrix}
				e^{-i\theta\frac{n+1-k}{n+1} H_{n,k} } &  \\
				& e^{-i\theta \frac{k}{n+1} H_{B, n,k-1}} e^{-i\theta\frac{n+1-k}{n+1}H_{n,k} } \\
			\end{matrix} \right]. 
		\end{equation*}
		
		As $n$ increases, the extra pair of Hadamard gate $H^{(n+1)}$ applied on a single qubit cannot influence the magnitude of the eigenspectrum of the whole Hamiltonian. This is straightforward in a large-scale quantum system that the application of a single-qubit gate does not significantly alter the eigenspectrum of the overall Hamiltonian, as the system's energy levels are predominantly governed by the multi-qubit interactions. Consequently, our focus shifts to the eigendecomposition of $U'_0 = U_{H_1, n+1, k}(\theta)H^{(n+1)}U_{H_2, n+1, k}(\theta)H^{(n+1)}$, represented in block-matrix form as 
		\begin{equation*}
			U'_0 = \left[ \begin{matrix}
				U_{00}  &  \\
				& U_{11}  \\
			\end{matrix} \right],
		\end{equation*}
		where 
		\begin{align*}
			& U_{00} = e^{-i\theta\frac{n+1-k}{n+1} H_{B, n,k} } e^{-i\theta\frac{n+1-k}{n+1} H_{n,k} },  \\
			& U_{11} = e^{-i\theta\frac{n+1-k}{n+1}  H_{B, n,k} } e^{-i\theta \frac{k}{n+1}  H_{n,k-1}} e^{-i \theta\frac{k}{n+1} H_{B, n,k-1}}  e^{-i\theta\frac{n+1-k}{n+1}H_{n,k} }. \\
		\end{align*}
		Actually, $U'_0$ has an identical magnitude in eigenvalue of the Hamiltonian with that of 
		\begin{equation*}
			\mathcal{U}'_0 = \left[ \begin{matrix}
				e^{-i\theta\frac{n+1-k}{n+1}\mathcal{H}_{n,k} }&  \\
				& 	e^{-i\theta \left( \frac{n+1-k}{n+1}\mathcal{H}_{n,k} + \frac{k}{n+1}\mathcal{H}_{n,k-1}  \right) }   \\
			\end{matrix} \right]
		\end{equation*}
		according to the Trotter-Suzuki decomposition. 
		Thus, the spectral gap of Hamiltonian of $U_{n+1,k}(\theta)$ is in the same magnitude with $\mathcal{U}'_0.$
		
		As $n$ increases, the extra pair of Hadamard gates $H^{(n+1)}$ applied to a single qubit cannot affect the magnitude of the eigenspectrum of the entire Hamiltonian. This is straightforward in a large-scale quantum system that the application of a single-qubit gate does not significantly alter the eigenspectrum of the overall Hamiltonian, as the system's energy levels are predominantly governed by the multi-qubit interactions. Consequently, our analysis focuses on the eigendecomposition of $U'_0 = U_{H_1, n+1, k}(\theta)H^{(n+1)}U_{H_2, n+1, k}(\theta)H^{(n+1)}$, which can be represented in block-matrix form as 
		\begin{equation*}
			U'_0 = \left[ \begin{matrix}
				U_{00}  &  \\
				& U_{11}  \\
			\end{matrix} \right],
		\end{equation*}
		where 
		\begin{align*}
			& U_{00} = e^{-i\theta\frac{n+1-k}{n+1} H_{B, n,k} } e^{-i\theta\frac{n+1-k}{n+1} H_{n,k} },  \\
			& U_{11} = e^{-i\theta\frac{n+1-k}{n+1}  H_{B, n,k} } e^{-i\theta \frac{k}{n+1}  H_{n,k-1}} e^{-i \theta\frac{k}{n+1} H_{B, n,k-1}}  e^{-i\theta\frac{n+1-k}{n+1}H_{n,k} }. \\
		\end{align*}
		In fact, the Hamiltonian of $U'_0$ has an identical eigenspectrum to that of 
		\begin{equation*}
			\mathcal{U}'_0 = \left[ \begin{matrix}
				e^{-i\theta\frac{n+1-k}{n+1}\mathcal{H}_{n,k} }&  \\
				& 	e^{-i\theta \left( \frac{n+1-k}{n+1}\mathcal{H}_{n,k} + \frac{k}{n+1}\mathcal{H}_{n,k-1}  \right) }   \\
			\end{matrix} \right]
		\end{equation*}
		according to the Trotter-Suzuki decomposition. Thus, the spectral gap of the Hamiltonian associated with $U_{n+1,k}(\theta)$ is of the same magnitude as that of $\mathcal{U}'_0$.
	\end{proof}
	
	
	\begin{lem}\label{lemma_gap} 
		For a small constant $k$, the spectral gap $g_k$ of $\mathcal{H}_{n,k}$ scales as $\Theta(n^{-1})$. 
	\end{lem}
	\begin{proof}
		When $k=1$, it follows from Lemma \ref{lemma_eigen_1_local_QS} that the spectral gap is of the order $\Theta(n^{-1})$. For a small constant $k\ge2$ and $n$ that is not large but still sufficiently greater than $k$, this lemma can be readily verified. Supposing that this lemma holds for case of $k-1$ with any $n$ sufficiently larger than  $k-1$, and also for case of $k$ with a specific $n$ that is moderately large compared to $k$, we now aim to establish its validity for the case of $k$ with $n+1$.
		
		According to Lemma \ref{lemma_eigen_reduction}, the spectral gap of $\mathcal{H}_{n,k}$ is on the same order as that of $\tilde{\mathcal{H}}_{n+1,k}$. The spectral gap of $\tilde{\mathcal{H}}_{n+1,k}$ can be determined by analyzing each block matrix of $\tilde{\mathcal{H}}_{n+1,k}$, as shown in Eq.~(\ref{eq_approxi_Hn}). 
		The upper block matrix contains an extra coefficient $\frac{n+1-k}{(n+1)}$ comparing to the original Hamiltonian $\mathcal{H}_{n,k}$, resulting in smaller eigenvalues. For the lower block matrix $\frac{n+1-k}{n+1}\mathcal{H}_{n,k} + \frac{k}{n+1}\mathcal{H}_{n,k-1}$, recalling the deviation of the objective function in Eq.~(\ref{eq_deviation_fkx}), it can be represented as $\mathcal{H}_{n,k} + \frac{k(n-d)}{(n+1)(n-k+1)}\mathcal{H}_{n,k-1}$. Denoting this deviation of Hamiltonian as $\Delta \mathcal{H}_0$, every element of $\Delta \mathcal{H}_0$ is on the order of $\mathcal{O}(n^{-2})$. Consequently, the spectral gap of $\tilde{\mathcal{H}}_{n+1,k}$ is determined by the smaller of the spectral gap of lower block matrix and the gap between the maximal eigenvalues of the block matrices. 
		
		Given that the lower block matrix $\frac{n+1-k}{n+1}\mathcal{H}_{n,k} + \frac{k}{n+1}\mathcal{H}_{n,k-1}$ deviates from the main component $\mathcal{H}_{n,k}$ by a positive semidefinite Hamiltonian of order $\mathcal{O}(n^{-2})$, its eigendecomposition should closely resemble that of $\mathcal{H}_{n,k}$. Consequently, the spectral gap of this block matrix is expected to remain on the order of $\mathcal{O}(n^{-1})$. Regarding the gap between the maximal eigenvalues of the both block matrices, the extra coefficient $\frac{n+1-k}{(n+1)}$ in the upper block matrix introduces a significant gap, also of order $\mathcal{O}(n^{-1})$. Thus, this lemma holds for $k$ with $n+1$.
	\end{proof}
	
	\begin{lem}\label{lemma_k_local_QS}
		With $\theta = \Theta(n^{-1})$, ${\mathcal O}(n^2)$ iterations are necessary for $k$-local quantum search to evolve the state such that the amplitude of $\left| t \right>$ reaches its first local maximum.
	\end{lem}
	\begin{proof}
		Assuming the eigendecomposition of $\mathcal{H}_{n,k}$ is given by $V^\dagger E V$, the evolution of $e^{-i \theta p_\theta \mathcal{H}_{n,k}}$ can be viewed as a rotation within the eigenspace defined by $V$. Let $\left| \tilde{t} \right>$ be the eigenstate with the eigenvalue closest to that of the target state $\left| t \right>$. Given a spectral gap of order $\Theta(n^{-1})$ in $\mathcal{H}_{n,k}$ between $\left| t \right>$ and $\left| \tilde{t} \right>$, if the number of iterations $p_\theta$ exceeding $\Theta(n^2)$, the phase difference in the evolution between $\left| t \right>$ and $\left| \tilde{t} \right>$ become larger than $\mathcal{O}(1)$, i.e., $\theta p_\theta g_{n,k} > \mathcal{O}(1)$. Due to the periodicity nature of the phase, this leads to a significant over-rotation, resulting in phase differences between different computational basis states that are distributed chaotically over the range $\left[ 0, 2\pi \right]$. In this scenario, the amplitude of the target state cannot exhibit a monotonic increase over this number of iterations. Therefore, the required number of iterations should be on the order of $\mathcal{O}(n^2)$.
	\end{proof}
	

	\section{Example of algorithms} \label{apsec_alg}
	This section presents example algorithms used to solve the problems discussed in this paper. Algorithm~\ref{alg_classical_sol_k_local} provides a classical solution to the $k$-local search problem. When $k$ is held constant, it requires $\mathcal{O}(n)$ Oracle calls. Algorithm~\ref{alg_main} outlines the quantum solution for the max-$k$-SSAT problem, which utilizes the $k$-local adiabatic quantum search routine $\text{AQS}_k(H_C, T)$ to handle the general instances and Grover's quantum search routine $\text{QS}_n(H_C)$ for extreme cases, where $T$ denotes the evolution time for adiabatic computation. This algorithm is proven efficient as long as the number of clauses $m = \Omega(n^{1+\epsilon})$, as demonstrated in the proof of the main theorem.
	
	\begin{algorithm}[H]
		\small
		\caption{Classical solution to $k$-local search problems}
		\KwData{Objective function $f_k(x)$;}
		\KwResult{target $t$;}
		\label{alg_classical_sol_k_local}
		$x \gets 0$\;
		\While{$f_k(x) = 0$} {$x \gets \text{RandomInt}(0, N-1)$}
		\For{$j \gets 1$ \KwTo $n$ }{
			$x' \gets x$\;
			$x'_j \gets \lnot x_j$\;
			\If{$f_k(x) < f_k(x')$}{$x\gets x'$\;}
		}
	\end{algorithm}
	
	\begin{algorithm}[H]
		\small
		\caption{Quantum search algorithm to max-$k$-SSAT}
		\KwData{Boolean formula $\Phi \in F_s(n,m,k)$ and its problem Hamiltonian $H_C$;}
		\KwResult{Interpretation $t$;}
		\label{alg_main}
		$T \gets n^2$\;
		$t \gets \text{AQS}_k(H_C, T)$\;
		\While{$ \lnot\Phi(t) \land (T \le \sqrt{N})$} {
			$T \gets 2T$\;
			$t \gets \text{AQS}_k(H_C, T)$\;
		}
		\If{$\lnot\Phi(t)$}{$t\gets \text{QS}_n(H_C)$\;}
	\end{algorithm}
	
	
	
	\bibliography{STOC_QSk}
	\bibliographystyle{alphaurl}
\end{document}